\def\year{2022}\relax
\documentclass[letterpaper]{article} 
\usepackage{aaai22}  
\usepackage{times}  
\usepackage{helvet}  
\usepackage{courier}  
\usepackage[hyphens]{url}  
\usepackage{natbib}  
\usepackage{caption} 
\DeclareCaptionStyle{ruled}{labelfont=normalfont,labelsep=colon,strut=off} 
\usepackage{algorithm}
\usepackage{amsmath, bm, amsfonts}
\usepackage{layouts}
\usepackage{lipsum}

\usepackage{newfloat}
\usepackage{listings}
\floatstyle{ruled}
\newfloat{listing}{tb}{lst}{}
\floatname{listing}{Listing}

\nocopyright
\usepackage{graphicx}
\usepackage{siunitx}
\usepackage[font=small]{caption} 
\usepackage{amsthm}
\usepackage{layouts}

\newtheorem{prop}{Proposition}
\newtheorem{lemma}{Lemma}

\usepackage{lipsum}
\usepackage{xcolor}
\usepackage{subcaption}
\usepackage{algpseudocode}
\usepackage{tikz}
\usepackage{xcolor}
\definecolor{green}{rgb}{0.35,0.8,0.0}
\definecolor{blue}{rgb}{0, 0.4470, 0.7410}
\usetikzlibrary{matrix, positioning, shapes,decorations.pathreplacing,fit,backgrounds}
\title{Separable Shape Tensors for Aerodynamic Design}

\author {
    Zachary J. Grey,\textsuperscript{\rm 1}
    Olga A. Doronina, \textsuperscript{\rm 2}
    Andrew Glaws \textsuperscript{\rm 2}
}
\affiliations {
    \textsuperscript{\rm 1} National Institute of Standards and Technology, Boulder, CO, USA\\
    \textsuperscript{\rm 2} National Renewable Energy Laboratory, Golden, CO, USA\\
    zachary.grey@nist.gov, olga.doronina@nrel.gov, andrew.glaws@nrel.gov
    }

\begin{document}
\maketitle
\begin{abstract}
Airfoil shape design is a classical problem in engineering and manufacturing. In this work, we combine principled physics-based considerations for the shape design problem with modern computational techniques using a data-driven approach. Modern and traditional analyses of 2D and 3D aerodynamic shapes reveal a flow-based sensitivity to specific deformations that can be represented generally by affine transformations (rotation, scaling, shearing, translation). We present a novel representation of shapes that decouples affine-style deformations over a submanifold and a product submanifold principally of the Grassmannian. As an analytic generative model, the separable representation, informed by a database of physically relevant airfoils, offers (i) a rich set of novel 2D airfoil deformations not previously captured in the data, (ii) an improved low-dimensional parameter domain for inferential statistics informing design/manufacturing, and (iii) consistent 3D blade representation and perturbation over a sequence of nominal 2D shapes.
\end{abstract}
\section{Introduction}
We begin by reviewing aspects of airfoil and blade/wing design to establish a motivation for the work. What follows is intended to serve as a detailed overview of the theoretical foundations for computations. Implementations and examples are available on GitHub~\cite{doecode_73484}.

Many artificial intelligence (AI)-aided design and manufacturing algorithms rely on shape representation methods to manipulate shapes in order to study sensitivities, approximate inverse problems, and inform optimizations. Two-dimensional cross sections of aerodynamic structures such as aircraft wings or wind turbine blades, also known as \emph{airfoils}, are critical engineering shapes whose design and manufacturing can have significant impacts on the aerospace and energy industries. Research into AI and machine learning (ML) algorithms involving airfoil design for improved aerodynamic, structural, and acoustic performance is a rapidly growing area of work~\cite{seshadri2018turbomachinery, Zhang:2018,Li:2019,Chen:2019,Glaws:2022a,Glaws:2022b,Wang:2022,Yonekura:2021,Yang:2022}.

Although airfoil shapes have been studied extensively and can appear relatively benign, their representation and design are complex due to their extreme operating conditions and the highly sensitive relationship between shape deformations and changes in aerodynamic performance. In this context, innovations specifically related to computational domains are of paramount importance for the future of computational fluid dynamics~\cite{slotnick2014cfd}. Improved shape parameter domains will enable future parametrized model reductions~\cite{willcox2002balanced, benner2015survey} to balance computational costs, improve designs, and make computations more explainable and interpretable. 

In this work, we explore a data-driven approach that uses a matrix ($2$-tensor) manifold framework to parametrize (or learn) a manifold of airfoil shapes. The resulting set of deformations to airfoil shapes separates important, and often constrained, affine deformations. Modern airfoil design incorporates constrained design characteristics of twist (i.e., angle of attack) and scale, which must be fixed or treated independently of higher-order deformations to the shape. Our approach decouples these two aspects of airfoil design and offers new interpretations of a space of shapes not previously considered---that is, ``learning'' a manifold of discrete shapes as submanifolds built from parent matrix manifolds. In the following subsections, we provide a brief overview of the airfoil representation scheme and demonstrate its flexibility over current methods, including the capability to extend from two-dimensional (2D) airfoils to full three-dimensional (3D) shapes, such as wind turbine blades. The results are predicated on parametrizations over well-understood manifolds, offering an \emph{analytic generative model} for airfoil shapes, in contrast to alternative AI-based generative models and other nonlinear dimension reductions (i.e., manifold learning). Implementations and examples are available~\cite{doecode_73484}.

\subsection{Defining an Airfoil}
We review general concepts for airfoil design, highlighting the fact that airfoil shapes are defined independent of planar rotations and scaling despite these deformations being highly sensitive to aerodynamic quantities of interest. This presents a challenge to define a general domain of airfoil shapes which is also independent of a chosen basis expansion representation.

Airfoil design seeks an optimal planar shape that satisfies desired design criteria---e.g., a specific lift and drag under particular operating and atmospheric conditions. Quantitatively, a given airfoil is typically characterized by its aerodynamic properties using lift and drag profiles, or ``polars,'' which are defined as univariate functions of these quantities over a range of planar rotations (or angles of attack). These profiles characterize important operational behavior and sensitivities over a continuous set of planar rotations. As such, functionals of these univariate profiles are often used to inform quantities of interest for optimization, approximation, and inverse problems---e.g., $\int(\zeta \circ f)(\theta)d\theta$, where $f$ is lift, drag, or the ratio of lift to drag as a function of scalar rotation angle $\theta$; and $\zeta$ acts as a transformation representing the desired characteristics of the polar. Specifying functionals over a planar rotation of the shape to characterize the airfoil operational profile suggests that a given airfoil is best characterized independently of these rotations. The rigid motion is ``integrated away'' by the definition of operational effectiveness as a functional over the angle of attack.

To study shapes, one approach is to represent airfoils by the class-shape transformation (CST) ~\cite{kulfan2008universal}, which partitions the shape into upper (suction side) and lower (pressure side) surfaces of the airfoil. The upper and lower surfaces are parametrized by the coefficients of two truncated Bernstein polynomial series. The two distinct upper and lower surface polynomials are multiplied by rational ``class functions'' that ensure at least two roots at zero and one where the distinct curves (as graphs of the two polynomials) connect. Consequently, deformations applied by modifying the coefficients of the expansion retain \emph{fixed orientation} by design, such that the leading edge point (root at zero) and trailing edge point (root at one) of the polynomial expansions are fixed. 

Alternatively, the airfoil can be parametrized using planar Bezier splines, or general B-splines, with specified control points~\cite{Hosseini2016}. Additionally, shapes can be inferred from noisy data with weighted local least squares~\cite{Ghorbani2021}. The curve is typically parametrized with complicated (and often subjective) constraints on each control point to help regularize deformations to achieve certain behavior. Previous studies have used splines to design several families of airfoils for modern megawatt-scale wind turbines~\cite{Li:2019}. These representations are typically fixed to ignore rotations of the shape and scaled to obtain a \emph{unit-chord length} such that the Euclidean distance from the \emph{leading edge} to the \emph{trailing edge} is one.

However, the discussed expansion representations (and others) are often \emph{coupled with a highly sensitive linear scaling of the shape} as inferred from the physics and modeling~\cite{Grey2017, Glaws:2022a, Li:2019, seshadri2018turbomachinery}. Specifically, general representations of shapes as curves remain coupled to large-scale, affine-type deformations---deformations resulting in significant and relatively well-understood physical impacts on aerodynamic performance like changes in thickness and camber. In contrast, smaller-scale undulating perturbations are of increasing interest to airfoil design problems~\cite{Glaws:2022a} and to the study of impacts of manufacturing defects and damage~\cite{Ge:2019}. This coupling between physically meaningful affine deformations and higher-order perturbations in shapes \emph{confounds deformations of interest} in the design process. Further details and examples of this coupling between deformations of interest and affine deformations are presented in Section~\ref{subsec:affine}.

Moreover, defining appropriate domains and constraints informing meaningful design spaces for the coefficients of Bernstein polynomials or B-splines can be very challenging from one problem to the next. Defaulting to bounded ranges about nominal coefficients may not be expressive enough to cover more diverse classes of airfoils. Restrictive and complicated choices of parameter domains make it challenging to diversify designs and take full advantage of the flexibility offered by AI-aided design and manufacturing. That is, interpreting the dependencies between perturbations to CST or spline parameters and subsequent changes to the resulting shapes can be challenging. This makes it difficult to select informed prior distributions generating ``reasonable'' aerodynamic shapes that underpin fundamental AI and inverse problem computations. For example, assumed uniform and Gaussian distributions over shape parameters often translate to complex, nonintuitive distributions over any space of shapes---distributions which may be multi-modal and consisting of several disjoint clusters of shapes.

Constraints imposed by the chosen representation, e.g., CST or splines, tend to fix the orientation of the airfoil in the plane. The intuition is that airfoil shapes are \emph{characterized by perturbations that are independent of rotation}, given design functionals defined over profiles of rotation angle. The desired 2D rotational invariance and physics-based interpretations motivate our development of a novel separable representation of shapes---independent of a chosen basis expansion representation---for next-generation aerodynamic design.

\subsection{Defining a Blade}
We describe challenges associated with defining a 3D blade/wing design by interpolating a sequence of 2D shapes. The 2D cross sections are designed and constrained by affine deformations encoding scale, rotation, and position properties that are often fixed by structural constraints or legal regulations. We establish a clear need for a general framework which can accomplish interpolation of 2D shapes independent of prescribed affine deformations. Additionally, new representations should be free from specific 2D basis expansion representations since total parameter count scales poorly for 3D design.

As a natural extension of 2D airfoils, 3D aerodynamic design considers the construction of a blade or wing from landmark airfoils with specific operational characteristics along the length of the blade. Blades and wings are generally represented by 2D cross sections extruded along a spanwise axis in 3D---e.g., wind turbine blades or gas turbine blades~\cite{Hosseini2016}. In contrast to the 2D airfoil design problem, the relative orientation, scaling, and position of the airfoil shapes along the span is heavily coupled to the final design of the full 3D shape---which has structural and/or regulatory implications in the specific case of wind turbines. For example, the 2D airfoils are often scaled to achieve an appropriate Reynolds number within the 3D blade, per modeled flow conditions. These shapes are then carefully rotated and translated smoothly along the span axis---imparting a twist and bend to the blade---to achieve specific operational characteristics from the hub to the tip of the 3D design. This procedure amounts to interpolating a sequence of 2D airfoils to define the 3D blade. 

The proposed design procedure is challenging with existing airfoil parametrizations as the total number of parameters defining the blade scales with the number of 2D cross sections and parameters may change dramatically from one airfoil cross section to the next. \emph{It is not clear how conventional, potentially high dimensional, parametrizations can be interpolated from one 2D shape to the next to retain the sought affine characteristics of the blade}.

The methods presented here transform existing airfoil definitions in order to (i) inform rotation- and reflection-invariant designs of 2D airfoil shapes and (ii) extend 2D designs to 3D blades by enabling spanwise airfoil interpolations that decouple blade-shape perturbations from specified scalings and rotations. The key characteristic in both the 2D and 3D design tasks is that we seek \emph{separability} between airfoil deformations that scale and rotate the shape---defined by affine deformations---and those that introduce local, high-order undulations in the surface. Introducing the separability then offers designers the ability to independently select the types of deformations that are most important for their application. 

Further, we seek to accomplish these design tasks \emph{free from any specific expansion representation} involving CST or splines---i.e., in a manner that does not parametrize the form of a specific basis expansion. Lastly, we explore a concept of \emph{consistently deforming} the airfoils defining a 3D blade such that the total parameter dimensionality is independent of the number of provided 2D cross sections. This leads to a novel framework for the design of next-generation aerodynamic shapes over a principled choice of reduced dimension domain.

\subsection{Contributions}
This work formally develops the use of specific matrix manifolds as underlying parent topologies for defining an affine and expansion independent parametrization of a space of discrete shapes. In general, this reveals a systematic approach to ``learning'' a submanifold of shapes through a vector space of $n$-by-$2$ full-rank matrices. Our contributions include:
\begin{itemize}
    \item The definition of separable forms of discrete shapes independent of affine deformations and designed expansions for 2D and 3D design
    \item A metric space of discrete shapes with an improved notion of distance between shapes over a novel data-driven domain definition
    \item Detailed computational routines for learning a manifold-valued analytic generative model of shapes in seconds
    \item Definition of a 3D blade/wing using affine-independent interpolation of designed 2D airfoil shapes
    \item A novel approach for 3D design involving consistent blade/wing deformations---minimizing parameter count and deforming the blade in an intuitive way
\end{itemize}
We note that these novel interpretations applied to airfoil design are closely related to the pioneering seminal work of David G. Kendall in 1977 summarized in short~\cite{kendall1989} and elaborated in detail~\cite{kendall2009shape}. As such, modern treatments of discrete shapes are commonly referred to as Kendall shape spaces.

\section{Discrete Representation \& Deformation}
We introduce shapes defined by smooth curves inducing discrete shapes as matrices that are independent of a choice of basis expansion. In detail, we discuss deformations to and standardization of these discrete shapes using simple linear algebra decompositions. We also analyze the convergence of data preprocessing and provide comparisons to AI/ML generative models. Finally, we cover detailed Riemannianian interpretations informing computations necessary to build implementations from scratch.

We begin by developing the discretized representations of 2D airfoil shapes as the foundation for data-driven parametrizations that will facilitate 2D and 3D aerodynamic design. The challenge of interest in modern aerospace shape design is to define shape deformations that are independent of the well-studied and aerodynamically sensitive \emph{affine deformations} to a shape. Affine deformations over planar coordinate axes, typically parametrized as scaling and rotation, can greatly affect aerodynamic performance (lift and drag profiles) of an airfoil. Additionally, scaling is not limited to simply increasing the shape volume but can also include independent vertical or horizontal scaling as well as the shearing the shape's image in the plane. Thus, it is important to decouple affine deformations from higher-order deformations as \emph{undulations}\footnote{In this complementary context, we define an undulation in the shape as any remaining deformation that are not represented by linear deformations---i.e., informally considered higher-order or nonlinear variations in the shape.} and study them independently in 2D design. Additionally, we require geometric representations that preserve these often carefully chosen affine characteristics for 3D design.

A 2D shape can be represented as a boundary defined by the open (i.e., injective or one-to-one) or closed (i.e., injective, except at endpoints) curve ${\bm{c}:\mathcal{I} \subset \mathbb{R} \rightarrow \mathbb{R}^2:s \mapsto \bm{c}(s)}$, where $\mathcal{I}$ is a compact domain. Without loss of generality, we can assume that $\mathcal{I} = [0,1]$. In practice, we consider a discrete representation the 2D airfoil shape as an ordered sequence of $n$ \emph{landmarks} $(\bm{x}_i) \in \mathbb{R}^2$ for $i=1,\dots,n$ and $n\geq 3$ along the curve $\bm{c}(s)$. That is, we have landmark points ${\bm{x}_i = \bm{c}(s_i)}$ for ${0 \leq s_1 < s_2 <\dots < s_n \leq 1}$. Moving along the curve, this sequence of planar vectors defining the airfoil shape results in the matrix ${X = (\bm{x}_1, \dots, \bm{x}_n )^\top \in \mathbb{R}_*^{n \times 2}}$, where ${\mathbb{R}_*^{n \times 2}}$ refers to the set of 
full-rank $n \times 2$ matrices (i.e., the noncompact Stiefel manifold). The full-rank restriction ensures that we do not consider degenerate $X$ as a feasible discrete representation of an airfoil shape.

Pivotal in this analysis is that \emph{we do not require a choice of basis expansion as a representation to deform $\bm{c}$}. Instead, we assume we have access to a database of discrete shapes generated by a diverse set of potentially different representations defining a variety of airfoil shapes. Then, we work with statistics of discrete shapes as opposed to statistics associated with coefficients in an expansion.

The innovative characteristic of the proposed approach is representing discrete airfoil shapes as elements of a Grassmann manifold (or Grassmannian) $\mathcal{G}(n, 2)$ paired with a corresponding affine transformation defined by an invertible $2\times2$ matrix plus a translation. This definition of the airfoil shape gives rise to a \emph{separable representation}, making important subsets of deformations independent and allowing designers to make interpretable, systematic changes to airfoil shapes over either type of deformation. For example, one may seek to preserve average airfoil scale characteristics while independently studying all remaining undulating deformations as perturbations over the Grassmannian. Further, this separability enables the extension of 2D shape parametrizations to 3D blade and wing shapes in a low-dimensional and consistent manner.

\subsection{Affine Deformations} \label{subsec:affine}
We discuss affine deformations describing large-scale planar deformations to shapes---namely, the linear term. We present linear deformations that are known to vary aerodynamic performance significantly and/or require specific control informed by design constraints. This motivates the need to separate linear deformations in shape representations from higher order oscillations or undulations in the shapes.

Affine deformations (i.e., scale, rotation, shear, and translation) of an airfoil have the form $M^{\top}\bm{c}(s) + \bm{b}$, where $M \in GL_2$ is an element from the set of all invertible $2\times 2$ matrices\footnote{For brevity, we simply refer to $GL_2(\mathbb{R})$ as $GL_2$ given all data and computation is over the reals.} and $\bm{b} \in \mathbb{R}^2$. Note that deformations by rank-deficient $M$ would collapse the shape to a line or to a single point and are not considered physically relevant as they have zero area.

A shape represented by associated boundary $\bm{c}:\mathcal{I}\rightarrow \mathbb{R}^2$ can be expressed generally as $
\bm{c}(s) = (c_1(s), c_2(s))^{\top}
$ such that the univariate functions $c_i$ for $i=1,2$ admit basis expansions $c_i(s) = \sum_{k \in \mathcal{K}} a_{ki}B_{ik}(s)$ as B-splines of fixed order, B\'ezier curves, or otherwise. Setting $(\Phi(s))_{ik} = (B_{ik}(s))$ and $(A)_{ki} = a_{ki}$, we can write this through the lens of linear algebra as $\bm{c}(s) = \sum_{j=1}^2\bm{e}_j(\bm{e}_j^{\top}\Phi(s)A \bm{e}_j)$ where $\bm{e}_j$ is the $j$-th column of the $2$-by-$2$ identity matrix, $A$ is a $\vert\mathcal{K}\vert$-by-$2$ matrix of the coefficients parametrizing the curve, and $\Phi(s)$ a $2$-by-$\vert\mathcal{K}\vert$ matrix of evaluated basis functions. By varying $A$, we can deform the shape as $D(A)$ through some mapping $D(\cdot)$ returning a matrix of appropriate dimension and rank. However, this mapping is unknown in general and could vary dramatically from one representation $A$ to the next. Applying $M$ to the curve as a \emph{scaling/rotation/shearing} deformation then modifies $D(A)$ and vice versa.

As an example, a common choice is to take $c_1(s) = s$ as the identity with all other basis functions in $\Phi$ set to zero along this $i=1$ horizontal coordinate. Consequently, the vertical direction is parametrized along the horizontal planar-coordinate axis by $s$. Any deformation $D(A)$ to coefficients of the expansion deform the curve as $\bm{\hat{c}}(s) = \sum_{j=1}^2\bm{e}_j(\bm{e}_j^{\top}\Phi(s)D(A) \bm{e}_j)$ in the planar-vertical direction which is equivalent to the \emph{graph} of the deformed function $\hat{c}_2$ over a potentially scaled domain---i.e., $(\hat{a}_{11}s, \sum_{k} \hat{a}_{k2}B_{2k}(s))$ for $\hat{a}_{11} > 0$ avoiding reflection and maintaining the rank of $\hat{A}$ provided at least one $\hat{a}_{k2} \neq 0$ for $k=2,\dots,\vert\mathcal{K}\vert$. Any deformation of the curve---parametrized in any manner as $D$---is then realized as a new expansion, $\hat{A} = D(A)$. However, applying a linear scaling $M$ modifies the expansion and vice versa thus coupling deformation types as
\begin{align*}
    \bm{\hat{c}}(s) &= 
M^{\top}\left(\begin{matrix}
\hat{a}_{11}s,\\ \sum_{k} \hat{a}_{k2}B_{2k}(s)
\end{matrix}\right) \\&=\left(\begin{matrix}
M_{11}\hat{a}_{11}s + M_{21}\sum_{k} \hat{a}_{k2}B_{k2}(s)\\ M_{12}\hat{a}_{11}s+ M_{22}\sum_{k} \hat{a}_{k2}B_{k2}(s)\end{matrix}\right)
\end{align*}
where $M_{ij}$ are the corresponding entries of $M$. Even in the special case that off-diagonal elements of $M$ are zero (no shearing) and $M_{11} = 1$, $M_{22}$ must be non-zero to maintain $M\in GL_2$ but $M_{22}$ is modified by factoring any non-zero common denominator from the coefficients in the expansion of $\hat{c}_2$---i.e., generally, $D(A)$ represents the effect of a vertical scaling by a common factor. Thus, changing $\hat{A}$ changes scaling $M$ and changing $M$ changes our expansion $\hat{A}$ despite the systematic choice of expansion along the horizontal planar-coordinate---which is coupled to $M_{22}$ at a minimum. Of course, this is further complicated if $M$ constitutes a rotation or shearing. In other words, notice that the specific form of $\Phi$ does not decouple linear deformations $M$ applied to the shape represented by deforming $A$ via unknown $D(\cdot)$. However, specific choices of $D(\cdot)$ may accomplish scale-invariance subject to corresponding constraints on any given representation---e.g., restricting a function space to a sphere~\cite{hagwood2013testing}. 

As opposed to working with specific choices of $\Phi$ and constrained $D(\cdot)$, we opt to work with discrete shapes $X$ \emph{independent of chosen expansion representations}. Additionally, we have no prudent notion of distance in a high dimensional space containing $A$ and $\hat{A}$ which complicates selection of domain definitions when a Euclidean distance is presumed insufficient.

For a discrete shape representation, affine deformations can be written as the smooth right action with translation ${XM \, + \, 1_{n,2}\text{diag}(\bm{b})}$, where $1_{n,2}$ denotes the $n \times 2$ matrix of ones and $\text{diag}(\bm{b}) = \sum_{j=1}^2(\bm{e}_j^{\top}\bm{b})\bm{e}_j\bm{e}_j^{\top}$. Note that the translation of the shape $\bm{b}$ does not change the intrinsic characteristics of the shape (i.e., it has no deforming effect) and is generally of little interest for 2D design problems. For 3D blade design, $\bm{b}$ locates the landmark airfoils relative to one another and can define the center of rotation. 

Previous work on sensitivity analysis of CST parameters representing airfoil shapes has revealed certain shape deformations can dramatically change the coefficients of lift and drag~\cite{Grey2017, Glaws:2022a}. These deformations are similar to affine deformations of simultaneously changing camber and thickness---a result consistent with laminar flow theory. This dominating influence of affine deformations on aerodynamic quantities of interest inhibits the nuanced study of a richer set of perturbations to airfoil shapes, which is becoming increasingly important to continued progress in aerodynamics research. For example, the set of ``dents'' and ``dings'' common to damage and manufacturing defects---e.g., leading edge erosion and soiling of an airfoil shape---cannot be described entirely by affine deformations. However, a fundamental understanding of the impact of these features on aerodynamic performance can lead to increased longevity of expensive and difficult-to-replace components such as offshore wind turbine blades. This motivates the need for a set of parameters that describe deformations independent of those in this dominating class of affine transformations. More precisely, we seek transformations to separately treat smooth right actions over $GL_2$. This line of research was initially proposed as an extension of~\cite{Grey2017} in~\cite{grey2019active}.

Affine deformations constitute only a subset of the possible important aerodynamic deformations. We contend that aerodynamics will be significantly influenced by any parametrization, composition, or generalization of scaling/rotation so long as $M \in GL_2$. Moreover, these affine deformations are critical for 3D design and are usually constrained or rigorously chosen when selecting nominal definitions of shapes. For example, a useful example parametrization of the linear term is
\begin{align} \label{eq:eg_GL2}
    L_4:\,&\mathcal{L} \subset \mathbb{R}^4 \rightarrow GL_2\\ \nonumber
        &\bm{\ell} \mapsto \ell_1\left[\begin{matrix}
            \ell_2 & 0 \\
            0 & \ell_3\\
        \end{matrix}\right]\left[\begin{matrix}
            \cos(\ell_4) & \sin(\ell_4)\\
	        -\sin(\ell_4) & \cos(\ell_4)
        \end{matrix}\right].
\end{align}
This parametrization is representative of the types of systematic deformations chosen or constrained in blade design. That is, we compute volumetric ($\ell_1$) and coordinate-aligned horizontal ($\ell_2$) and vertical ($\ell_3)$ scalings, then rotate the shape into the final angle of attack ($\ell_4$) for assembly or modeling. Although \eqref{eq:eg_GL2} is not necessarily a common parametrization of $GL_2$, assuming all of $GL_2$ is aerodynamically significant offers more flexibility for designers to select or fix arbitrary deformations over $GL_2$ beyond those parametrized by \eqref{eq:eg_GL2} that may be deemed interesting.

We seek to decouple and preserve affine features for blade and wing design through a set of inferred shape deformations over the Grassmannian that are independent of $GL_2$. We also discuss separable shape tensors for parametrizing scale variations independent of rotation/reflection for individual airfoil (2D) design.

\subsection{Separable Shape Tensors}
We introduce the Grassmannian $\mathcal{G}(n,q)$ as a topology where variations in discrete shapes due to linear deformations are ``divided out.'' We describe how to map shapes to representative elements of the Grassmannian as Landmark-Affine (LA) standardizations using:  (i) the singular value decomposition and (ii) the related polar decomposition. This motivates parametrizations of sections through product manifolds:  (i) $\mathcal{G}(n,2) \times GL_2$ and (ii) $\mathcal{G}(n,2) \times S_{++}^2$, respectively. These will define the ``parent'' topologies for submanifolds of separable shape tensors.


\subsubsection{Landmark-Affine Standardizations}
\begin{figure}
    \centering\includegraphics[width=0.95\linewidth]{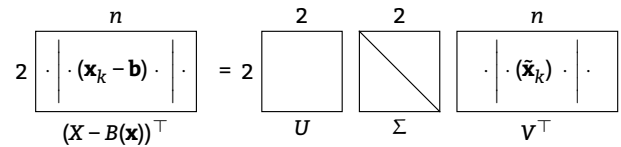}
    \caption{Thin singular value decomposition (SVD) of centered landmarks with appropriate dimensions and corresponding notation. Lines motivate an intuition for a relevant partitioning of any given matrix.}
    \label{fig:SVD}
\end{figure}
\begin{figure}
	\centering\includegraphics[width=0.95\linewidth]{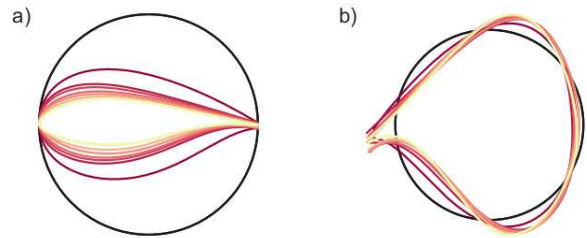}
	\caption{Collection of ten nominal cross-sectional airfoils defining the International Energy Agency (IEA) 15-MW blade in (a) physical coordinates, $X$, and (b) Landmark-Affine standardized coordinates, $\tilde{X}$.}
	\label{fig:affine_transform}
\end{figure}

Given a discrete shape representation $X$ and an important affine deformation $XM + 1_{n,2}\text{diag}(\bm{b})$, we now develop the necessary interpretations for deforming shapes independent of the often constrained and notably aerodynamically sensitive affine features. Through this development, we reveal underlying parent matrix manifold topologies which will inform improved non-Euclidean considerations for subsequent computations and 3D blade interpolation.

The Grassmannian $\mathcal{G}(n,q)$ is the space of all $q$-dimensional subspaces of $\mathbb{R}^n$. Formally, ${\mathcal{G}(n,q) \cong \mathbb{R}^{n\times q}_*/GL_q}$, meaning that elements of the Grassmannian are invariant under $GL_q$ transformations where $GL_q$ is the set of all invertible $q \times q$ matrices. Given this invariance, we may consider an element of the Grassmannian ${[\tilde{X}] \in \mathcal{G}(n,q)}$ to be the equivalence class of all matrices with the same column span as the representative element $\tilde{X} \in \mathbb{R}^{n \times q}_*$~\cite{absil2008optimization}. That is, the equivalence class ${[A] = \left\lbrace B \in \mathbb{R}^{n\times q}_* \,:\, B \sim A \right\rbrace}$ is defined by equivalence relation ${A \sim B}$ such that ${\text{Range}(A) = \text{Range}(B)}$. In this way, every element of the Grassmannian is a full-rank matrix modulo $GL_q$ deformations. Thus, deformations over $\mathcal{G}(n,q)$ are independent of affine deformations (ignoring the nondeforming translations)---i.e., $[XM] = [X]$. By representing discretized airfoil shapes as elements of the Grassmannian, we ensure that deformations to shapes or differences between shapes in this space are, by definition, decoupled from the aerodynamically important affine deformations---e.g., linear transformations varying camber, thickness, and length.

It is common to view the Grassmannian as a quotient topology of orthogonal subgroups such that the $n$ landmarks of any representative element $\tilde{X}$ have sample covariance proportional to the identity matrix---i.e., $\tilde{X}^\top\tilde{X} = I_q$~\cite{edelman1998geometry, gallivan2003efficient}. In practice, this means that a representative computational element of the Grassmannian is an $n \times q$ matrix with orthonormal columns~\cite{edelman1998geometry}. This perspective offers certain computational advantages and motivates a scaling of airfoil landmark data for computations over $\mathcal{G}(n,2)$ for airfoil design. In our case, $n$ is equal to the number of landmarks, and $q = 2$ is the dimension of the ambient space where the shape lives for design.

To represent physical airfoil shapes as elements of the Grassmannian, we define the Landmark-Affine (LA) standardization~\cite{bryner20142d} as a mapping ${\pi:\mathbb{R}^{n\times 2}_* \rightarrow \mathcal{G}(n,2)}$. LA standardization normalizes the shape to have zero sample mean (translation invariance) and identity sample covariance (scale invariance) over the $n$ landmarks defining the shape. The remainder of this section discusses computation of the LA standardization and examines its properties.

Given a discrete airfoil shape with landmarks $X \in \mathbb{R}_*^{n \times 2}$, let ${\bm{b}(X) = (1/n) \, X^{\top}1_{n,1}}$ be the discrete center of mass and compute the thin singular value decomposition (SVD) of the centered airfoil ${(X - B(X))^{\top} = U\Sigma V^{\top}}$, where ${B(X) = 1_{n,2}\,\text{diag}(\bm{b}(X))}$. Then, define $\tilde{M}$ to be the $2\times2$ invertible matrix
\begin{equation} \label{eq:LA_M}
    \tilde{M} = \Sigma U^{\top}.
\end{equation}
The mapping between the airfoil $X$ and its LA-standardized representation $\tilde{X}$ is
\begin{equation} \label{eq:LA_standardization}
    X - B(X) = \tilde{X} \tilde{M}.
\end{equation}
As a result, this definition of $\tilde{X}$ provides scale standardization
\begin{align*}
    \tilde{X}^{\top}\tilde{X} &= \tilde{M}^{-\top}(X-B(X))^{\top}(X-B(X))\tilde{M}^{-1} \\
    &= \tilde{M}^{-\top}(U\Sigma V^{\top})(V \Sigma U^{\top})\tilde{M}^{-1} \\
    &= \Sigma^{-1}U^{\top}(U\Sigma^2U^{\top}) U\Sigma^{-1}\\
    &= (\Sigma^{-1}\Sigma)(\Sigma\Sigma^{-1})\\
    &= I_2,
\end{align*}
consistent with a \emph{whitening transform}~\cite{hyvarinen2000independent}. From \eqref{eq:LA_standardization}, we have $\tilde{X} = V$, with standardized landmarks along the rows as $\tilde{X} = (\tilde{\bm{x}}_1,\dots,\tilde{\bm{x}}_n)^{\top} \in \mathbb{R}^{n\times 2}_*$. To clarify, Fig.~\ref{fig:SVD} depicts the dimensionality of the various matrices and the notation. The LA standardization \eqref{eq:LA_standardization} satisfies assumptions to apply various intrinsic parametrizations for computing normal coordinates over the Grassmannian~\cite{edelman1998geometry}.

For $[\tilde{X}] \in \mathcal{G}(n,2)$, $\tilde{X}$ is a \emph{representative} (Stiefel) element of the Grassmannian, defined uniquely up to any $GL_2$ transformation~\cite{edelman1998geometry, absil2008optimization}. Abstractly, we map a given discrete airfoil shape $X$ to an equivalence class $[\tilde{X}]$ via $\pi:\mathbb{R}_*^{n\times 2} \rightarrow \mathcal{G}(n,2):X \mapsto [\tilde{X}]$ such that
\begin{equation}
    \pi(X) = [\tilde{X}] = [X - B(X)] .
\end{equation}
Next, we show that $\pi(X) = [\tilde{X}]$ is surjective, thus admitting a parametrizable \emph{right inverse}, and satisfies the desired scale and translation invariance properties. Lastly, we show that $\pi$ is the canonical projection thus the right inverse parametrizes sections through $\mathbb{R}^{n\times 2}_*$ as a submanifold. The results establish a principled framework for ``learning'' a (sub)manifold of discrete shapes in $\mathbb{R}^{n\times 2}_*$ with the desired separability.

\begin{prop} \label{prop:surjective}
$\pi$ is surjective.
\end{prop}
\begin{proof}
For any $[X] \in \mathcal{G}(n,2)$, we can take an arbitrary basis $[X] = \text{span}\left\lbrace \bm{v}_1, \bm{v}_2\right\rbrace$ for linearly independent $\bm{v}_1, \bm{v}_2 \in \mathbb{R}^n$. Consequently, taking arbitrary $(a\bm{v}_1, b\bm{v}_2) \in \mathbb{R}^{n\times 2}_*$ for $a,b\neq 0$ implies $\pi((a\bm{v}_1, b\bm{v}_2)) = [X]$.
\end{proof}

\begin{prop} \label{prop:scale_invar}
$\pi$ is scale invariant such that $\pi(XM) = \pi(X)$ for any $M \in GL_2$.
\end{prop}
\begin{proof}
Defining $B(X) = 1_{n,2}\text{diag}(\bm{b}(X))$ where $\bm{b}(X) = (1/n)X^{\top}1_{n,1}$,
\begin{align*}
    \pi(XM) &= [XM - B(XM)]\\
    &= [XM - 1_{n,2}\text{diag}((1/n)(XM)^{\top}1_{n,1})]\\
    &= [XM - 1_{n,2}\text{diag}(M^{\top}\bm{b}(X))]\\
    &= [XM - 1_{n,2}\text{diag}(\bm{b}(X)^{\top}M)]\\
    &= [XM - (1_{n,2}\text{diag}(\bm{b}(X)))M]\\
    &= [(X - B(X))M]\\
    &= [X - B(X)]\\
    &=\pi(X)
\end{align*}
\end{proof}

\begin{prop} \label{prop:trans_invar}
$\pi$ is translation invariant such that ${\pi(X + 1_{n,2}\text{diag}(\bm{b}')) = \pi(X)}$ for any $\bm{b}' \in \mathbb{R}^2$.
\end{prop}
\begin{proof}
For arbitrary $\bm{b}' \in \mathbb{R}^2$,
\begin{equation*}
\pi(X + 1_{n,2}\text{diag}(\bm{b}')) = [X + 1_{n,2}\text{diag}(\bm{b}') - B(X + 1_{n,2}\text{diag}(\bm{b}')) ] . 
\end{equation*}
Then, we recognize that
\begin{align*}
B(X + 1_{n,2}\text{diag}(\bm{b}'))
=& 1_{n,2} \text{diag} ((1/n)(X + 1_{n,2}\text{diag}(\bm{b}'))^\top1_{n,1}) \\
=& 1_{n,2} \text{diag} ((1/n)X^\top1_{n,1} \\&+ (1/n)\text{diag}(\bm{b}')1_{n,2}^\top1_{n,1}) \\
=& 1_{n,2} \text{diag} ((1/n)X^\top1_{n,1} + \text{diag}(\bm{b}')1_{2,1}) \\
=& 1_{n,2} \text{diag} ((1/n)X^\top1_{n,1}) \\&+ 1_{n,2} \text{diag} (\text{diag}(\bm{b}')1_{2,1}) \\
=& B(X) + 1_{n,2} \text{diag}(\bm{b}') .
\end{align*}
Noting that $1_{n,2}^{\top}1_{n,1} = n \, 1_{2,1}$, $\text{diag}(\bm{b}') 1_{2,1} = \bm{b}'$, and $\text{diag}(\bm{u}+\bm{b}') = \text{diag}(\bm{u}) + \text{diag}(\bm{b}')$. Plugging this result into the equation above yields
\begin{align*}
\pi(X + 1_{n,2}\text{diag}(\bm{b}')) 
&= [X + 1_{n,2}\text{diag}(\bm{b}') - B(X + 1_{n,2}\text{diag}(\bm{b}')) ] \\
&= [X + 1_{n,2}\text{diag}(\bm{b}') - B(X) - 1_{n,2} \text{diag}(\bm{b}') ] \\
&= [X - B(X)] \\
&= \pi(X) . \\
\end{align*}
\end{proof}

\begin{prop} \label{prop:proj}
$\pi$ is the canonical projection
\end{prop}
\begin{proof}
It is sufficient to show that $\pi$ is idempotent onto equivalence classes, $\pi(\pi(X)) = \pi(X)$. For $\pi(X) = [\tilde{X}]$ and $(X - B(X))^{\top} = U\Sigma V^{\top}$, representative $\tilde{X} = (X-B(X))U\Sigma^{-1}$ has zero mean over rows, $B(\tilde{X}) = \bm{0}$, and $\tilde{X}^{\top}= \Sigma^{-1}U^{\top}U\Sigma V^{\top} = V^{\top}$, which at most rotates and/or reflects the shape after LA standardization informed by the thin SVD $V^{\top} = \tilde{U}I_2\tilde{V}^{\top}$. Consequently, $\pi(\pi(X)) = [\tilde{X}\tilde{U}] = [\tilde{X}] = \pi(X)$.
\end{proof}

Prop.~\ref{prop:scale_invar} motivates an alternative interpretation of $\pi$ as a $GL_2$ scale invariance, $\pi(XM) = [\tilde{X}]$. Intuitively, $\pi(X)$ \emph{standardizes} the discretized shape $X$ such that $\tilde{X}$ is as circular as possible. Fig.~\ref{fig:affine_transform} depicts a set of example transformations between these two discrete representations for a collection of wind turbine airfoil shapes. Prop.~\ref{prop:scale_invar} together with Prop.~\ref{prop:trans_invar} assert the sought affine invariance properties of a nonlinear mapping\footnote{The thin SVD from one discrete shape to the next is computed by an iterative procedure which is, in general, nonlinear over a space of changing matrices.} onto the Grassmannian, per Prop.~\ref{prop:surjective}. With Prop.~\ref{prop:surjective} and Prop.~\ref{prop:proj}, we can propose the parametrization of a \emph{section} through $\mathbb{R}^{n\times 2}_*$ using representative elements $\tilde{X}$ to build a submanifold. Consequently, we can define a separable shape tensor parametrization for discrete shapes,
\begin{equation} \label{eq:sep_shape}
    X(\bm{t}, \bm{\ell}) = (\pi^{-1}\circ [\tilde{X}])(\bm{t},\bm{\ell}) = \tilde{X}(\bm{t})M(\bm{\ell}).
\end{equation}
In this formalism, $\pi^{-1}$ is the right inverse of $\pi$ parametrized by ${(\bm{t},\bm{\ell}) \in \mathcal{T} \times \mathcal{L} \subseteq \mathbb{R}^{r + 4}}$. This parameter domain will be inferred from data-driven pairs $\lbrace [\tilde{X}_k], M_k\rbrace \subset \mathcal{G}(n,2) \times GL_2$ given by the thin SVD of discrete shapes, $\lbrace X_k \rbrace$. In practice, $M(\bm{\ell})$ could be expressed for design as the composition of a fixed nominal scaling with a parametrized affine deformation---e.g., $M(\bm{\ell}) = \overline{M}L_4(\bm{\ell})$ possibly with translations where $L_4(\bm{\ell})$ is defined in \eqref{eq:eg_GL2} and $\overline{M}$ is some fixed nominal scaling like an average. Fig.~\ref{fig:fiber} shows a simplified visual analogue of this approach for a constant average scale factor $\overline{M}$---we further elaborate on the ability to average over $GL_2$ via separability in the next section. The parametrization of the Grassmannian element $[\tilde{X}](\bm{t})$ is inferred from data-driven methods that are also discussed in later sections. We note that the dimension of $\bm{t}$ is restricted by the intrinsic dimensionality of $\mathcal{G}(n,2)$ such that $r \leq 2(n-2)$ but is practically chosen to be much smaller. 

\begin{figure}
    \centering
    \includegraphics[width=0.95\linewidth]{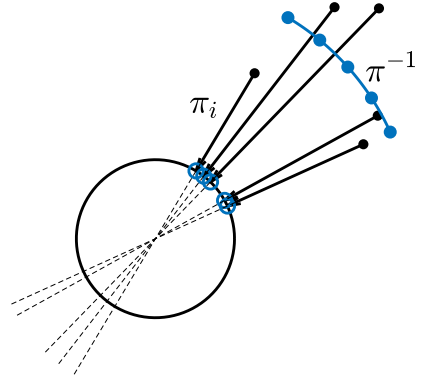}
    \caption{A simplified visualization of five individual projections $\pi_i$ (black arrows) from elements in $\mathbb{R}_*^{2\times 1}$ (black dots) onto representative elements of the upper semicircle (blue circles). Elements of $\mathcal{G}(2,1)$ are shown as dashed lines. A constant section of the fiber bundle (a submanifold of $\mathbb{R}_*^{2\times 1}$) is shown as the blue curve, with five (uniformly sampled) elements as coincident dots.}
    \label{fig:fiber}
\end{figure}

The utility of the representation in \eqref{eq:sep_shape} is the \emph{separable form} of the airfoil representation such that changes in $\bm{t}$ are independent of changes in $\bm{\ell}$. The lingering question is: Given a database of discrete airfoils $\left\lbrace X_k \right\rbrace$, \emph{how can we infer parameter distributions} of $(\bm{t},\bm{\ell})$? Alternatively, how should we define $\mathcal{T} \times \mathcal{L}$ for subsequent design tasks? The key will be utilizing data-driven approaches involving the underlying Riemannian geometry---described in the sections to follow.

\subsubsection{Mean Scales of Random Airfoils}

When considering an \emph{ensemble} of airfoil shapes $\lbrace X_k \rbrace$, an average notion of scale can be used to remove the dependencies on $M(\bm{\ell})$. For example, we could define the constant extrinsic estimate $\left. \overline{M} = (1/N)\sum_{k=1}^N \tilde{M}_{k} \right.$, where  $\tilde{M}_{k}$ is computed as in \eqref{eq:LA_M} for the corresponding $X_k$. Assuming these airfoils implicitly define some distribution over $\mathcal{L}$, this offers a notion of \emph{average scale} for parametrizing a \emph{local section of the fiber bundle} through \emph{total space} $\mathbb{R}^{n \times 2}_*$ as $(\pi^{-1} \circ [\tilde{X}])(\bm{t}; \overline{M})$~\cite{grey2019active}. Fig.~\ref{fig:fiber} depicts a simplified visual analogue for this choice of constant scaling---represented by the blue curve. 

When designed or inferred affine deformation subgroups, or arbitrary parametrizations such as \eqref{eq:eg_GL2}, are combined with translations, they can then be applied independently to $\tilde{X}(\bm{t})\overline{M}$ as a systematic design schema. As an aerodynamic interpretation, for unknown domain $\mathcal{L}$ weighted by unknown probability measure $\rho$, order-dependent compositions of camber, chord, twist, and/or thickness deformations can be independently applied to Monte Carlo approximations of average-scale shapes, ${\tilde{X}(\bm{t})\overline{M} \approx \int_{\mathcal{L}}\tilde{X}(\bm{t})M(\bm{\ell})d\rho(\bm{\ell}) = \tilde{X}(\bm{t})\int_{\mathcal{L}}M(\bm{\ell})d\rho(\bm{\ell})}$. The challenge is ensuring that $\overline{M} \in GL_2$ and that $\overline{M}$ does not arbitrarily inflate average length scales of the shape---i.e., $\overline{M}$ can result in an inflated determinant. This may require control through additional nonlinear transformations (e.g., a set of shape constraints) applied to $\tilde{X}(\bm{t})\overline{M}$ by the shape designer. Alternatively, we could pose an intrinsic mean scale over $GL_2$. In either case, the ability to average scales $\lbrace \tilde{M}_{k}\rbrace$---or separately compute higher-order moments of important and highly sensitive scale variations over alternative metric spaces---is enabled by the separability in \eqref{eq:sep_shape}. 

\subsubsection{Equivalent Polar Decomposition}
We next develop an alternative method for mapping discrete airfoil shapes to the Grassmannian based on a rotation-invariant polar decomposition. We denote the resulting product manifold of shapes as $\mathcal{G}(n,2) \times S_{++}^2$, where $S_{++}^2$ denotes the set of $2 \times 2$ symmetric positive definite (SPD) matrices. First, we define an equivalence relationship $\tilde{X} \sim_{\mathcal{O}} \tilde{X}O$ for all orthogonal $2 \times 2$ matrices $O\in \mathcal{O}(2)$ such that $[X]_{\mathcal{O}} = \left\lbrace A \in \mathbb{R}^{n \times 2}_* \,\,:\,\, A \sim_{\mathcal{O}} X\right\rbrace$. Taking the polar decomposition of the linear transformation as $M = P_2R_2$, we have $[\tilde{X}M]_{\mathcal{O}} = [\tilde{X}P_2R_2]_{\mathcal{O}} = [\tilde{X}P_2]_{\mathcal{O}}$. Thus, we can parametrize a set of rotation/reflection-invariant shapes as $\tilde{X}(\bm{t})P(\bm{\ell})$ for all $P(\bm{\ell})\in S_{++}^2$. That is, we retain the scale variations of shapes over SPD matrices and ``divide out'' deformations resulting in rotations and reflections.

Given a discrete shape $X$ and its corresponding thin SVD $\left. (X - B(X))^{\top} = U\Sigma V^{\top} \right.$, the polar decomposition becomes $\left. (X - B(X))^{\top} = PR \right.$ such that $P = U\Sigma U^{\top}$ is unique and $R = UV^{\top}$. Given $V$ as an $n \times 2$ rectangular matrix with orthonormal columns and $U$ orthogonal implies that $RR^{\top} = I_2$. Moreover, $P^{-1} = U\Sigma^{-1}U^{\top}$ is SPD. Hence, $\tilde{X} = (X - B(X))P^{-1}$ defines an equivalent normalization of scale,
\begin{align*}
    \tilde{X}^{\top}\tilde{X} &= P^{-1}(X - B(X))^{\top}(X - B(X))P^{-1} \\
    &= P^{-1}(PR)(R^{\top}P)P^{-1} \\
    &= (P^{-1}P)(PP^{-1})\\
    &= I_2.
\end{align*}
Equivalently, $\tilde{X} = R^{\top} = VU^{\top}$, which rotates/reflects the original nonunique LA-standardized shape $V$ back into the original view by $U^{\top}$. 

For an ensemble of shapes $\left\lbrace X_k \right\rbrace$, we can compute the corresponding $P_k$ from the approximated thin SVD and use the data-driven pairs to construct a submanifold from ${\lbrace [\tilde{X}_k], P_k\rbrace \subset \mathcal{G}(n,2) \times S^2_{++}}$---i.e., a set of rotation/reflection-invariant subgroups of discrete shapes. The separable form of rotation/reflection-invariant physical airfoils becomes
\begin{equation} \label{eq:polar_sep}
     X(\bm{t}, \bm{\ell}) = (\pi^{-1}\circ [\tilde{X}])(\bm{t},\bm{\ell}) = \tilde{X}(\bm{t})P(\bm{\ell})
\end{equation}
for parameters $(\bm{t}, \bm{\ell}) \in \mathcal{T} \times \mathcal{P} \subseteq \mathbb{R}^{r+3}$. Notice also that $P(\bm{\ell})$ serves to parametrize the eigenspaces of the landmark sample covariance, i.e., $(X - B(X))^{\top}(X - B(X)) = U\Sigma U^{\top} = P$. In other words, parametrizing linear scale variations in the shape is equivalent to parametrizing the sample covariance of landmarks, establishing the utility of examining and parametrizing the range of $U$.

\subsubsection{Convergence to Discrete Representatives}

\begin{figure}
    \centering
    \includegraphics[width=0.85\linewidth]{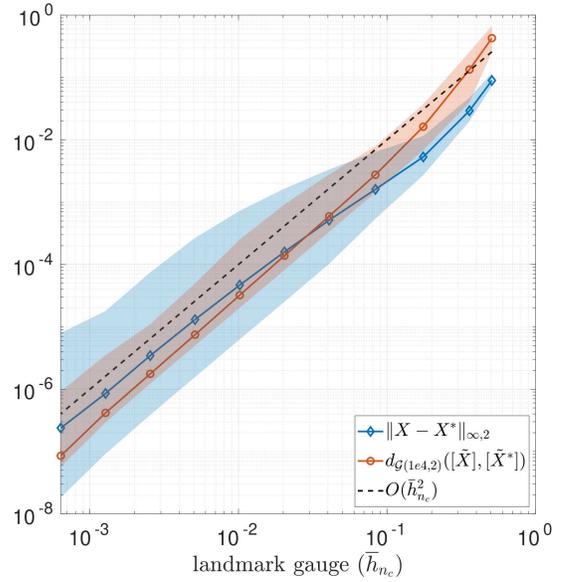}
    \caption{Convergence of $n=10,000$ refinements $X$ and $\tilde{X}$, generated by planar splines reparametrized over \eqref{eq:piecewise_arclengths}, given $n_c$ landmarks as data with corresponding landmark gauge $\overline{h}_{n_c}$. Shaded regions represent a max-min envelope over all $100$ random CST airfoils. Colored curves represent the average Euclidean shape distance and Grassmannian distance (sum of principal angles) between interpolated refinements and true refinements over all random shapes. The dashed black curve depicts the observed quadratic convergence.}
    \label{fig:cst_convergence} 
\end{figure}

The interpretation and use of the Grassmannian for the purpose of describing a topology of discrete shapes is a relatively unique application of the manifold. The motivation stems from the application-driven need for separable representations to study or systematically control distinct affine deformations of known physical importance. There are concerns about how this treatment of discrete shapes may change when subject to reparametrizations as diffeomorphisms, i.e., $\xi:\mathcal{I} \rightarrow \mathcal{I}$ defining new landmarks $\bm{x}_i = \bm{c}(\xi(s_i))$ that constitute the transposed rows of $X$. Modern definitions of shape spaces~\cite{welker2021suitable, Michor2006} are typically defined modulo such diffeomorphisms, and alternative frameworks~\cite{dogan2015fast, joshi2007novel, klassen2004analysis} take advantage of a \emph{pre-shape space} with \emph{square root velocity functions}, inducing a notion of distance between shapes.

In shape design, $\xi$ is often defined in an effort to best identify sequences $(\bm{x}_i)$ that are distributed along the arc-length of the shape with increased concentration around regions of high curvature or are distributed uniformly with high resolution (e.g., $n \geq 1000$ in airfoil design) to improve meshing in a flow solver. We propose fixing $\xi$ such that it induces a specific distribution over chosen parameter $s$ to generate corresponding landmark refinements, $X \in \mathbb{R}^{n\times 2}_*$.

Data sets of discrete shapes rarely share a common choice of generating landmark distribution, particularly if the shapes are gathered from multiple sources. Further, the number of landmarks in the discrete shapes $n_c$ may vary across the data set. We consider a simple \emph{preprocessing} of data before working with tensor representations. This preprocessing proceeds as follows. Given landmark data $(\bm{x}_i)$, we first compute normalized cumulative Euclidean lengths,
\begin{equation} \label{eq:piecewise_arclengths}
s_i = 
    \begin{cases}
        0, \,\, i=1
        \\
        \sum_{k=1}^{i-1}\Vert \bm{x}_{k+1} - \bm{x}_{k}\Vert_2 \Biggl/ \sum_{k=1}^{n_c-1}\Vert \bm{x}_{k+1} - \bm{x}_{k}\Vert_2,
    \end{cases}
\end{equation}
for $i=2,\dots,n_c$. We then interpolate the entries of $(\bm{x}_i)$ over $\xi(s_i)$ to construct a continuous representation $\hat{\bm{c}}(s)$ as an approximation. Finally, we generate a fixed $n$-\emph{refinement} $X \in \mathbb{R}_*^{n \times 2}$ as $n$ (typically greater than $n_c$) landmarks generated by interpolations and LA standardize the $n$-refinement to produce $\tilde{X} \in \mathbb{R}^{n\times 2}_*$. This procedure also offers control over the \emph{landmark gauge}, $\overline{h}_n = \underset{k\in\lbrace 1,\dots,n-1\rbrace}{\text{max}}\lbrace\Vert \bm{x}_{k+1} - \bm{x}_{k} \Vert_2\rbrace$, of the $n$-refinement for subsequent meshing and simulation. This preprocessing motivates the following result:

\begin{lemma} \label{prop:convergence}
Given sparse data $(\bm{x}_i)_{i=1}^{n_c}$ from assumed $\bm{c} \in C^4(\mathcal{I})$, $n$-refinements $X \in \mathbb{R}^{n\times 2}_*$ generated by cubic splines of $(\bm{x}_i)_{i=1}^{n_c}$ with fixed reparametrization $\xi$ converge to elements $[\tilde{X}] \in \mathcal{G}(n,2)$ as $O(h_{n_c}^{4})$ where
$
h_{n_c} = \underset{k\in\lbrace 1,\dots,n_c-1\rbrace}{\text{max}}\vert s_{k+1} - s_{k} \vert
$
is the gauge of an arc-length parametrization.
\end{lemma}

\begin{proof}
By the polar decomposition $(X - B(X))^{\top} = PR$, $P$ is unique and thus $\tilde{X} = R^{\top} = (X - B(X))P^{-1}$ is unique, provided $X$ is full rank, which is true by assumption. Consequently, $[\tilde{X}] = [R^{\top}]$ is unique. 

In the most general case, an interpolation $\tilde{\bm{c}}$ must be evaluated over an appropriately weighted domain to maintain a correspondence of landmarks consistent with a fixed reparametrization $\xi:\mathcal{I} \rightarrow \mathcal{I}$ over arc-lengths $s \in \mathcal{I}$. Otherwise, inconsistent landmarks are encoded in the rows of the $n$-refinement. We argue convergence over the worst-case error across rows of $X$ computed with $\hat{\bm{c}} = \tilde{\bm{c}} \circ \xi$, necessitating a choice and/or construction of fixed $\xi$ with preimage as arc-lengths $s \in \mathcal{I}$. We define
$$
\Vert \bm{e} \Vert_{\infty,2} = \underset{s}{\text{max}} \lbrace \Vert \bm{e}(s) \Vert_2\,:\,s\in \mathcal{I}\rbrace
$$
as the \emph{worst-case} Euclidean norm over any row encoded in a corresponding matrix. For fixed $\xi$ and $h_{n_c}$ as the gauge of the $n_c$-mesh, it follows that the worst-case row converges for decreasing $h_{n_c}$ by applying results from \cite{birkhoff1964error} (sharpened by \cite{hall1976optimal}), asserting
$$
\Vert \hat{\bm{c}} - \bm{c} \Vert_{\infty, 2} \leq K_4h_{n_c}^4
$$
which bounds an equivalent matrix norm applied to a corresponding difference in $n$-refinements. The aforementioned result derives constant $K_4$ which depends on the largest fourth derivatives of the component functions of $\bm{c}$.
\end{proof}

We experiment with the common CST representation of airfoils to generate $100$ random airfoil shapes according to a  cosine domain distribution, resulting in a nonuniform (over arc-length) distribution of landmarks.\footnote{The cosine distributed landmarks are a common choice for airfoil designers because landmarks concentrate around the ``leading edge'' and ``trailing edge'' features of an airfoil. In this case, it reflects an intuitive choice of nonuniform landmark distribution.} The CST representation utilizes $18$ total coefficients ($9$ upper surface coefficients and $9$ lower surface coefficients), each uniformly sampled over $[0,0.45]$. However, the CST shapes are represented by a partition into upper and lower surface thus potentially compromising the assumed differentiability. Moreover, in general, we often do not have true arc-lengths along the shape as data. Instead, we supplement by using \eqref{eq:piecewise_arclengths} as a parametrization of spline(s) with corresponding landmark gauge. We numerically study the effects of abusing the underlying assumptions of the theory to demonstrate empirically that our necessary preprocessing of shapes remains convergent.

Given a known CST shape and an appropriate (design-informed) sampling scheme, we compute a refinement at $n=10,000$, defining LA-standardized shape $\tilde{X}^*$ with corresponding cumulative Euclidean lengths $( s^*_i )_{i=1}^n$ inferred from $X^*$ (prior to LA standardization). Then, we sparsely resample the CST representation with the same sampling scheme for $n_c \leq n$
and provide the sparse landmarks to a cubic spline interpolation routine parametrized over sparse $n_c$ using \eqref{eq:piecewise_arclengths}. The spline approximation is then used to generate LA standardized refinements $\tilde{X}$ up to the corresponding $n=10,000$ according to $\xi(s^*_i)$---e.g., $\xi$ is the identity for the planar spline built from landmark entries over $( s_i )_{i=1}^{n_c}$, while alternative splines may utilize $\xi$ as a PCHIP such that $\tilde{s}_i = \xi(s_i)$ correspond to \eqref{eq:cummulative_LA_lengths} for the same $n_c$-mesh. The nature of the spline interpolation over cumulative Euclidean lengths as an approximation induces an error that manifests as a Euclidean shape distance, $\Vert X - X^* \Vert_{\infty,2} = \underset{k\in\lbrace 1,\dots,n\rbrace}{\text{max}}\lbrace\Vert \bm{x}_{k} - \bm{x}^*_{k} \Vert_2\rbrace$, or as the sum of principal angles between $n$-refinement $[\tilde{X}]$ and $[\tilde{X}^*]$ as a distance over the Grassmannian, $d_{\mathcal{G}(n,2)}([\tilde{X}],[\tilde{X}^*])$. Evidence is depicted in Fig.~\ref{fig:cst_convergence}. 

Increasing $n_c$ landmarks with reduced landmark gauge provided as input data improve accuracy in the separable tensor representations. The observed order of convergence---specifically for the chosen random CST airfoils as a particularly relevant class of shapes---to known representative elements on the Grassmannian is reduced from the result of the Lemma~\ref{prop:convergence}. However, we observe that the preprocessing procedure still offers a relatively fast order of convergence over a class of relevant shapes.

\subsubsection{Shortcomings}
The proposed separable representations \eqref{eq:sep_shape} and \eqref{eq:polar_sep} are not without drawbacks. In particular, physically relevant shapes $\bm{c}$ are nonintersecting, with the possible exception of closed curves, which necessarily coincide at the endpoints, i.e., embeddings~\cite{welker2021suitable}. Our current formalism could violate this. Specifically, the discrete separable representations $(\pi^{-1} \circ [\tilde{X}])(\bm{t},\cdot)$ can generate self-intersections in continuous reconstructions for \emph{large changes} in parameters $\bm{t}$. We mitigate these concerns by preventing extrapolation beyond embeddings (implicit to the data) using a numerical routine to check for a piecewise linear intersection condition (accurate for sufficiently large $n$). We hypothesize that constraining against significant extrapolation, with an improved notion of distance, beyond a set of discrete nonintersecting shapes as data is sufficient to protect against generating self-intersections---this is supplemented by empirical evidence. Future work will focus on improved constraints or continuous analogues avoiding self-intersection to explicitly constrain to a set of embedded curves.

\subsection{Continuous Analogues \& Comparisons}
We unpack an interpretation of the LA standardized shapes as they relate to orthogonal polynomials. Thus extending discrete representations to continuous forms. We conclude with comparisons to alternative AI/ML frameworks for generative modeling of discrete representatives. These discussions set the stage for future work and possible numerical comparisons.

We explore building continuous analogues from the discrete shapes by approximating a so-called \emph{quasimatrix}---i.e., a 2D array that is discrete along one dimension and continuous along the other~\cite{townsend2015continuous}. The basic idea is to explore a procedure for computing orthogonal functions given two discretizations (one per column) encoded in the rows of $V$. Specifically, we seek interpolations for the columns of the LA standardized shape $V = \tilde{X}$ (see proof of Prop.~\ref{prop:proj}) such that 
\begin{equation} \label{eq:discrete_kernel}
    (X - B(X))(X - B(X))^{\top}\bm{v}_j = \sigma_j^2\bm{v}_j, \qquad j = 1,2,
\end{equation}
where $\sigma_1,\sigma_2$ are the diagonal entries of $\Sigma$ in \eqref{eq:LA_M}. In this interpretation, \eqref{eq:discrete_kernel} is a discrete analogue of the Fredholm integral equation
\begin{equation} \label{eq:cont_kernel}
    \int_{\mathcal{I}}k_{\langle \bm{c}, \bm{c}\rangle}(s', s)v_j(s')d\mu(s') = \sigma_j^2v_j(s)
\end{equation}
for some measure $\mu$. 

We present \eqref{eq:discrete_kernel} to motivate an interpretation of the columns of $V$ as they relate to \eqref{eq:cont_kernel}. In this interpretation, the columns of $V$---$\bm{v}_j$ for $j=1,2$---represent discretizations of two $L^2(\mathcal{I})$ orthonormal eigenfunctions $v_j(s_i)$ for $s_i \in \mathcal{I}$ with ordering $s_1 < s_2 <\dots < s_n$. These eigenfunctions constitute coordinate functions for an LA-standardized planar curve $\tilde{\bm{c}}(s) = (v_1(s), v_2(s))$. The Mercer kernel $k_{\langle \bm{c}, \bm{c}\rangle}$ is the Euclidean inner product of the centered planar curve with itself, akin to entries of a Gram matrix but distinct from landmark sample covariance. This choice is consistent with the ``$A=0$'' metric discussed in supporting work~\cite{Schulz2014, Michor2006, joshi2007novel}. Although this choice of metric admits a pathology in the continuous framework~\cite{Michor2006}, it may still be worthwhile to consider in a discrete setting~\cite{Schulz2014, joshi2007novel}. This interpretation establishes the utility of examining and parametrizing the range of $V$. Additionally, this interpretation motivates how we may modify our implicit choice of kernel and/or shape metric for future applications.

Next, we inform continuous reconstructions using interpolation of discrete data while retaining the desired separability in deformations. To begin, we compute $V$ via the SVD of centered $X$. Then, we consider the normalized cumulative Euclidean length along the discrete curve $V = \tilde{X} = (\tilde{\bm{x}}_1,\dots, \tilde{\bm{x}}_n)^{\top}$, over which we will construct our interpolation,
\begin{equation} \label{eq:cummulative_LA_lengths}
    \tilde{s}_i = \begin{cases}
    0, \,\, i=1\\
    \sum_{k=1}^{i-1}\Vert \tilde{\bm{x}}_{k+1} - \tilde{\bm{x}}_{k}\Vert_2 \Biggl/ \sum_{k=1}^{n-1}\Vert \tilde{\bm{x}}_{k+1} - \tilde{\bm{x}}_{k}\Vert_2,\,\,
    \end{cases}
\end{equation}
for $i=2,\dots,n$. This results in pairs $\lbrace (\tilde{s}_i, v_{ij}) \rbrace$ for $j=1,2$, where $v_{ij}$ is the $ij$ entry of $V$. We can construct barycentric Lagrange interpolation over these data pairs, defined by
\begin{equation}
    v_j(\tilde{s}) = \sum_{i=1}^n \frac{w_i}{\tilde{s} - \tilde{s}_i}v_{ij} \Biggl/\sum_{i=1}^n \frac{w_i}{\tilde{s} - \tilde{s}_i},
\end{equation}
where weights are given by
$
w_i = 1/\prod_{k\,\neq\,i}(\tilde{s}_i - \tilde{s}_k)
$
for all $\left. i=1,\dots,n \right.$~\cite{berrut2004barycentric, higham2004numerical}. Alternatively, we could employ piecewise interpolating splines with conditions designed for improved \emph{fairness}~\cite{sapidis1994designing} to build the two curves. Alternative interpolations include regularized cubic splines (clamped, natural, or periodic), piecewise cubic Hermite interpolating polynomial (PCHIP) splines, B-splines, nonuniform rational basis splines (NURBS), Hicks-Henne bump functions, or radial basis functions. The results, in any case, are two functions $v_1(\tilde{s})$ and $v_2(\tilde{s})$, that interpolate pairs $\lbrace(\tilde{s}_i,v_{ij})\rbrace$ for all $i=1,\dots,n$ with $j=1$ or $j=2$, respectively. Concatenating $(v_1(\cdot), v_2(\cdot))$ into the columns of $V(\cdot)$, the interpolations defining $V(\cdot)$ are no longer necessarily orthogonal but retain some nice (often subjective, yet prescriptive) design characteristics. However, despite the choices defining the designed continuous reconstruction $V(\cdot)$, the two interpolations can be evaluated uniformly to induce a fixed reparametrized integral measure $d\tilde{\mu}(\tilde{s}) = d\tilde{s}$ along the curve, which can then be projected onto a space of orthonormal Legendre polynomials\footnote{In this case, the arbitrary closed interval $\mathcal{I}$ is reparametrized to $[-1,1]$ in contrast to the choice of $[0,1]$ in \eqref{eq:cummulative_LA_lengths}.} via a QR-decomposition of the prescribed $\infty \times 2$ quasimatrix~\cite{trefethen2010householder}---i.e., $V(\cdot) = Q(\cdot)R$ for $Q$ an $\infty \times 2$ quasimatrix and $R$ a $2\times 2$ upper triangular. The quasimatrix $Q(\cdot)$ from the QR-decomposition becomes the continuous analogue that satisfies the orthonormal constraint of the representative discrete shape $\tilde{X}$, which is defined by a relevant choice of LA standardized landmark data interpolation, $V(\cdot)$.

If desired, the Legendre polynomials as columns of $Q(\cdot)$---with sufficient differentiability over $\tilde{s}$---can be used to compute unit tangent and normal vectors of the airfoil shapes as well as curvature after right multiplication with corresponding scale variations. The continuous representation can then be expressed as
\begin{equation} \label{eq:GL2_continuous_analouge}
    C(\tilde{s};\bm{a}, \bm{\ell}) = Q(\tilde{s}; \bm{a})M(\bm{\ell})
\end{equation}
or
\begin{equation} \label{eq:SPD_continuous_analouge}
    C(\tilde{s};\bm{a}, \bm{\ell}) = Q(\tilde{s}; \bm{a})P(\bm{\ell})
\end{equation}
parametrized by a vector of polynomial coefficients $\bm{a}$ and scale (length) variations $\bm{\ell}$ over the respective choice of manifold, generally for $M(\bm{\ell}) \in GL_2$ or specifically for $P(\bm{\ell}) \in S_{++}^2$. In other words, an interpolation approximating the continuous reconstruction of the curve is $\tilde{\bm{c}}(\tilde{s}; \bm{a}, \bm{\ell}) = (q_1(\tilde{s}; \bm{a}_1, \bm{\ell}), q_2(\tilde{s}; \bm{a}_2, \bm{\ell}))$, where $q_i$ are parametrized linear combinations (over variations in $\bm{\ell}$) of two orthogonal Legendre polynomials constituting the columns of $Q(\cdot)$ with coefficients $\bm{a} = (\bm{a}_1, \bm{a}_2)$ over the continuous dimension of $C(\cdot; \bm{a}, \bm{\ell})$. In this context, $(\bm{a}_1, \bm{a}_2)$ represents a partition of the full set of coefficients $\bm{a}$ into corresponding component functions $(q_1, q_2)$ of the curve.

\subsubsection{AI/ML Comparisons and Interpretations}
We consider alternative approaches that leverage AI-based tools for dimension reduction and generative modeling, such as autoencoders or variational autoencoders (VAEs)~\cite{Kingma:2013,Kramer:1991} and generative adversarial networks (GANs)~\cite{Goodfellow:2014}. Autoencoder models learn nonlinear reduced representations by mapping input data through a so-called information bottleneck (i.e., the latent space) before reconstructing it. The basic architecture of these models is the composition of an encoder $\pi_{enc}$ with a decoder $\pi_{dec}$, with each component part comprised of multiple neural processing layers. Given a training data set $\lbrace X_k \rbrace$, we fit the model parameters $\bm{w} \in \mathbb{R}^D$ by minimizing some reconstruction loss
\begin{equation} \label{eq:VAE}
    \underset{\bm{w} \in \mathbb{R}^D}{\text{minimize}}\, \sum_{k}\Vert X_k - (\pi_{dec} \circ \pi_{enc})(X_k; \bm{w}) \Vert .
\end{equation}
VAEs expand on traditional autoencoders by encouraging desirable distributions on the latent variables by adding a term such as the Kullback-Leibler divergence $KL(\pi_{enc} (X_k; \bm{w}) || \rho(\bm{z}))$, where $\rho(\bm{z})$ is the target latent space distribution.

GANs are similarly comprised of two neural network models; however, unlike with autoencoders, these models are trained against each other. The generator network $\pi_{gen}$ maps random latent vectors $\bm{z} \sim \rho$ to outputs in the space of the training data. The discriminator network $\pi_{disc}$ maps data to $[0,1]$ corresponding to a probabilistic prediction that the input data came from the generator or the training data. The network parameters are fit according to a minimax optimization of the training data,
\begin{align} \label{eq:GAN}
    \underset{\bm{w}_{gen} \in \mathbb{R}^{D_{gen}}}{\text{minimize}}\,&\underset{\bm{w}_{disc} \in \mathbb{R}^{D_{disc}}}{\text{maximize}}\, \sum_{k} \log \pi_{disc} (X_k; \bm{w}_{disc}) \\
    &+ \sum_{k}\log \left( 1 - \pi_{disc} (\pi_{gen}(X_k; \bm{w}_{gen}); \bm{w}_{disc}) \right) . \nonumber
\end{align}
This optimization ultimately encourages the generator to draw plausible samples from the training data distribution. The generator from GANs and the decoder from VAEs can both be used to map low-dimensional, latent space parametrizations to new realizations. Both methods have been applied to the design of airfoil shapes~\cite{Yonekura:2021,Yang:2022,Chen:2019,Wang:2022,Achour:2020}.

In practice, we seek a robust and interpretable parametrization for the decoder/generator model to act upon. For VAEs, this corresponds to $\pi_{enc}$ being surjective onto the latent space. However, this property can be difficult to guarantee in general. In contrast, the proposed tensor representations are supported by Props.~\ref{prop:surjective}--\ref{prop:proj} and Lemma~\ref{prop:convergence}. The implications of ambiguous properties on the VAE latent space are ill-posed constructions and parametrizations that are highly dependent on the stochastic training process for the model. Consequently, the resulting parametrizations are often overfit to the specific data types and sets used during training. Although the targeted loss terms can encourage desirable behavior, these results are not guaranteed. Despite this, geometric interpretations have led to a set of boundary-value problems defined by the pullback metric (geodesic equations) to offer improved notions of distances between points in latent space~\cite{arvanitidis2017latent}. Such innovations and interpretations are crucial in the continued development of VAEs.

GANs benefit from a well-defined parameter space, as the distribution over the latent variables is decided prior to network training. However, it can be difficult to control the relationship between the latent variables and their resulting shape generations, leading to poor interpretability and little insight into the intrinsic dimension of the shape parametrization. Similar to the VAEs, the resulting parametrization is heavily influenced by the randomized network initialization and training procedure. Furthermore, although the adversarial training encourages the generation of quality shapes, issues such as mode collapse and overparametrization of the networks may cause the generator to miss key novel shape designs that drive innovation. Lastly, the computational burden---and the associated energy costs---required to train such sophisticated models is considerable~\cite{Strubell:2019}. 

In comparing these approaches to this work, the decoder $\pi_{dec}$ from VAEs and the generator $\pi_{gen}$ from GANs may be considered as analogues of the right inverse $\pi^{-1}$, which constitute a parametrization over a \emph{local section of the fiber bundle}~\cite{lee2006riemannian}. Exploring this connection further, the analogous network latent spaces will become \emph{normal coordinates}, defined in~\cite{lee2006riemannian}, over ``parent'' matrix manifolds that generate separable shape deformations in our context. These normal coordinates constitute a set of naturally defined parameters describing general manifold topologies, as opposed to the obscure latent space emulating a target distribution.

In Section~\ref{subsec:riemann_views}, we describe how our principled separable representations \eqref{eq:sep_shape} and \eqref{eq:polar_sep} offer improved geometric interpretations with significantly reduced computational cost. In particular, our principled approach to shape representation takes advantage of rigorously studied matrix manifolds and linear algebra to avoid the need for general nonconvex optimization and numerical integration of boundary-value problems when computing geodesics and distances over latent spaces. Additionally, we assert additional geometric interpretations beyond geodesics and distances that enable novel deformations of 3D shapes---a means of interpolating and applying consistent deformations to distinct 2D shapes. The result is an analytic generative model from a learned (data-driven) manifold of shapes.

\subsection{Riemannian Interpretations} \label{subsec:riemann_views}
We formally develop a data-driven framework for parametrizing elements over topologies of separable shape tensors, which leverage an extension of principal component analysis to Riemannian manifolds. This requires a pair of fundamental intrinsic maps for mapping between a given manifold and a tangent space at a central element. We also discuss an improved notion of distance as lengths of geodesic curves over the manifold. Lastly, we present parallel transport as a method for applying consistent deformations to different shapes---motivating a novel approach to deform 3D blades. These interpretations are backed by sections detailing algorithms to compute all necessary maps over the presented manifolds.

\begin{figure*}
     \centering
     \begin{subfigure}{0.32\textwidth}
         \centering
        \includegraphics[width=0.95\linewidth]{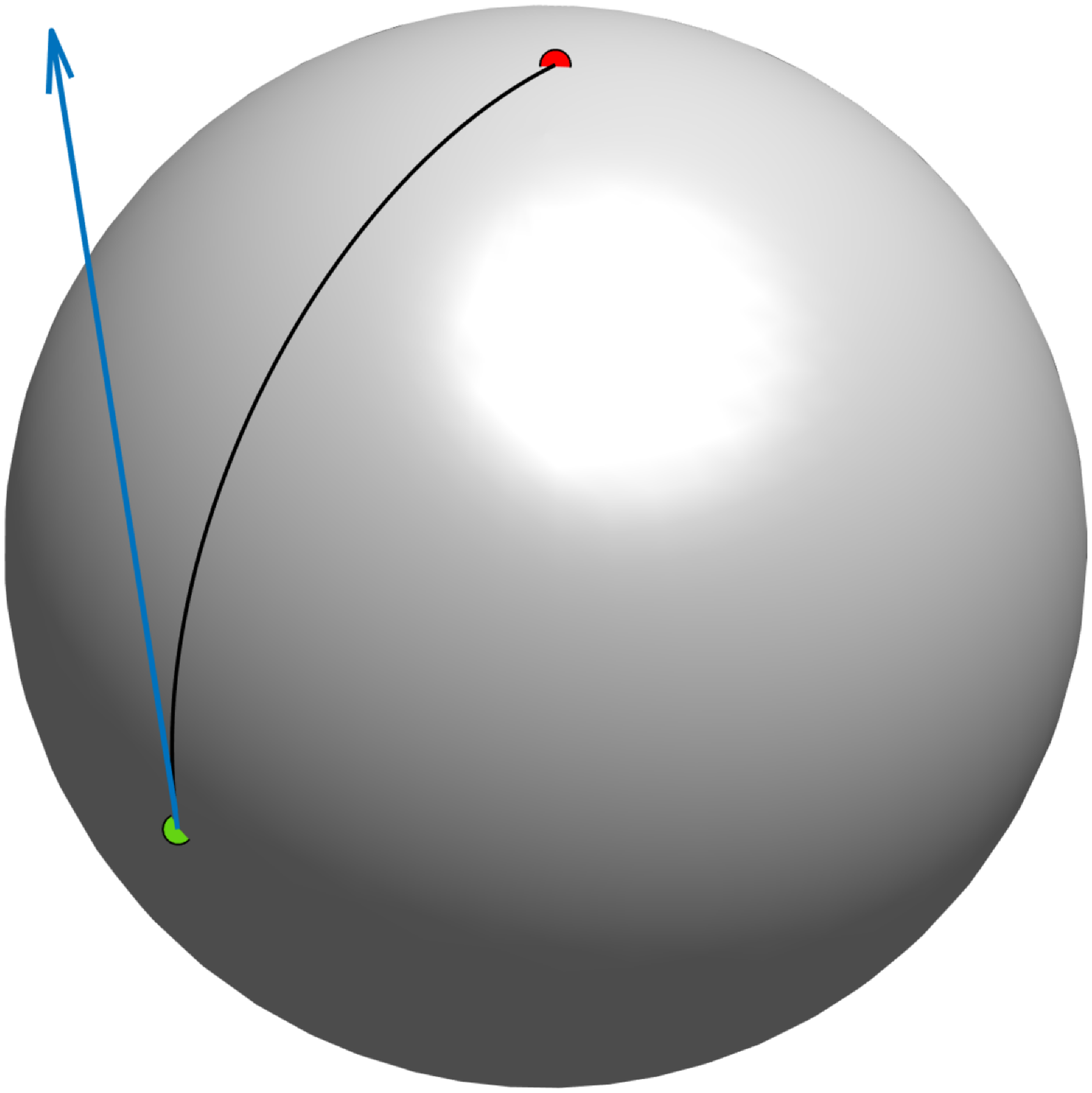}
        \caption{Intrinsic maps over the $2$-sphere}
        \label{fig:intrinsic_maps}
     \end{subfigure}
     \hfill
     \begin{subfigure}{0.32\textwidth}
         \centering
        \includegraphics[width=1\textwidth]{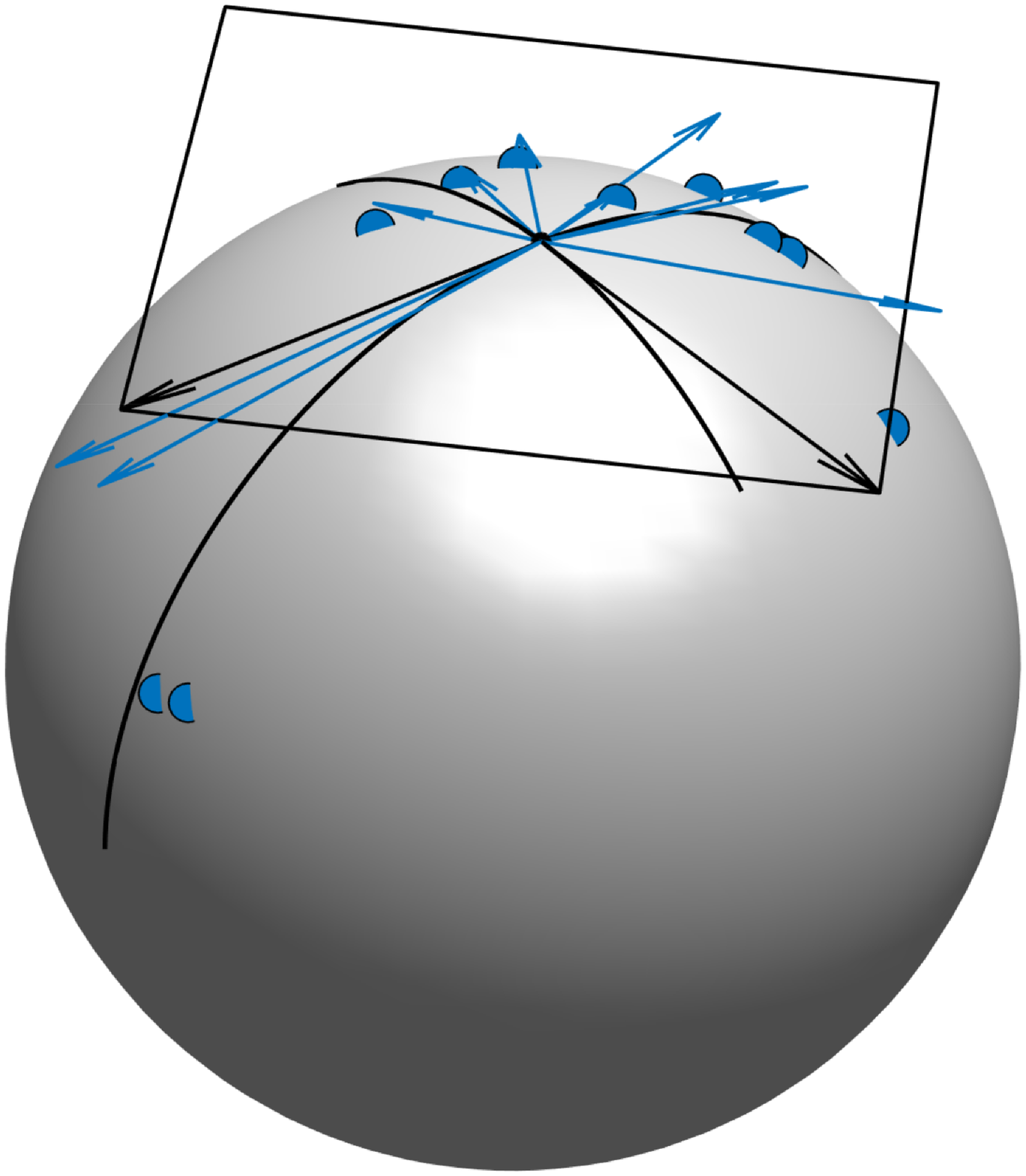}
        \caption{Principal geodesic analysis over the $2$-sphere}
        \label{fig:PGA_sphere}
     \end{subfigure}
     \hfill
     \begin{subfigure}{0.32\textwidth}
         \centering
        \includegraphics[width=1\textwidth]{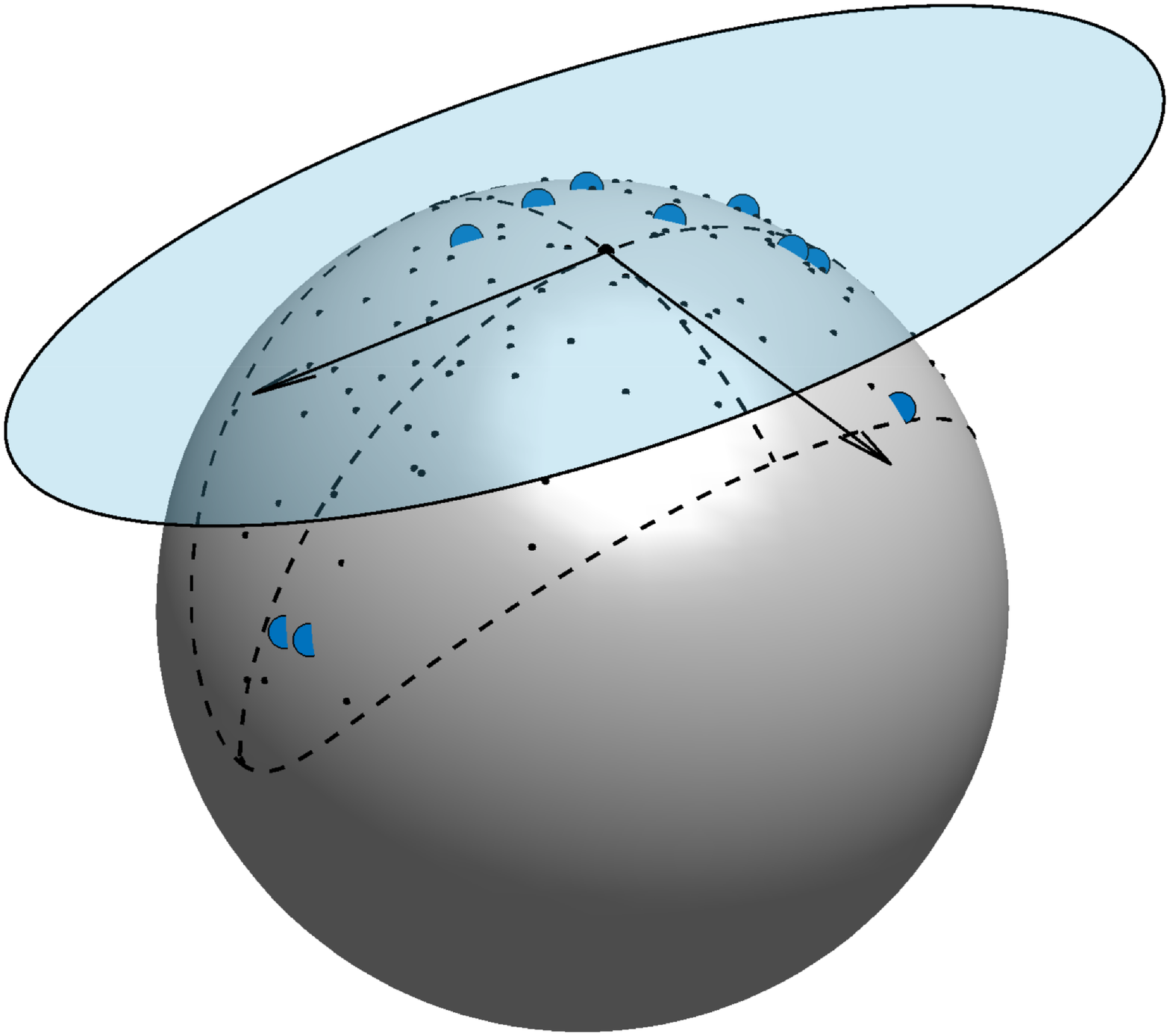}
        \caption{Geodesic ellipse over the $2$-sphere}
        \label{fig:nml_coords}
     \end{subfigure}
        \caption{Visualizations of data-driven methods over the $2$-sphere: (a) A visualization of the exponential $\text{Exp}_{\textcolor{green}{A}}(\textcolor{blue}{\Delta}) = \textcolor{red}{B}$ and inverse exponential $\text{Log}_{\textcolor{green}{A}}(\textcolor{red}{B}) = \textcolor{blue}{\Delta}$ maps over the sphere for points $\textcolor{green}{A},\,\textcolor{red}{B} \in \mathcal{M}$ and corresponding tangent vector $\textcolor{blue}{\Delta} \in T_{\textcolor{green}{A}}\mathcal{M}$. (b) Principal geodesic analysis (PGA) over the two-sphere embedded in $\mathbb{R}^3$. The large black dot represents the intrinsic mean of the blue dots, which are provided as data. The black arrows represent PGA directions in the central tangent space at the intrinsic mean. The black lines represent geodesics along the PGA directions. (c) A geodesic ellipse and grid (dashed boundary) lifted to the normal coordinate neighborhood (light blue region with solid boundary). Smaller black dots are new samples generated uniformly over the submanifold spanned by both directions---constituting a generative model.}
        \label{fig:Riemann_spheres}
\end{figure*}
In Section~\ref{subsec:affine}, we demonstrated how to define airfoil representations that separate important aerodynamic scale variations from higher-order undulations in the shape. The methods for LA standardization, namely polar decomposition, map discrete airfoil shapes $X$ to representative elements of the Grassmannian $\tilde{X}$. Here, we provide a data-driven approach to parametrizing these separable shapes that leverages rigorously designed data sets of airfoil coordinates~\cite{UIUC} and/or systematically engineered airfoil expansion representations~\cite{kulfan2008universal}. The goal is to use ensembles of existing, well-designed shapes $\lbrace X_k\rbrace$ to construct the separable forms in \eqref{eq:sep_shape} or \eqref{eq:polar_sep}. To do so, we perform a statistical analysis of a given ensemble of discrete shapes factored into the separable forms. This constitutes a data-driven perspective for reparametrizing shapes via separable tensors.

By introducing separability, we must consider the Riemannian geometry of $\mathcal{G}(n,2)$ and $S_{++}^2$, or possibly $GL_2$, to parametrize our discrete shape space. In general, we consider any as a smooth manifold $\mathcal{M}$ such that the Riemannian manifold is defined as $(\mathcal{M},g)$ for some choice of metric $g$ that induces an inner product $g_p:T_p\mathcal{M} \times T_p\mathcal{M} \rightarrow \mathbb{R}$ over the tangent space $T_p\mathcal{M}$ for some $p\in \mathcal{M}$. Note that in this abstracted sense, $p$ is a representative matrix element of the smooth manifold, whereas $\bm{v} \in T_p\mathcal{M}$ are generalized notions of \emph{directions} in a tangent space at $p$ but are also ultimately represented as matrices. 

Building submanifolds from parent matrix manifolds offers several advantages over alternative approaches, including modern AI/ML methods discussed previously. This perspective supplements a novel treatment for learning a non-Euclidean manifold of discrete shapes from data, namely by Prop.~\ref{prop:surjective} and Prop.~\ref{prop:proj}. This supplements a more prudent notion of distance between discrete shapes and their separable deformations. Shapes in this framework are defined by \emph{explainable} and \emph{interpretable} representations induced by parent matrix manifolds. That is, there is no obscurity about our explicit parametrizations defining submanifolds of discrete shapes---an argument which is very difficult to develop in general for alternative frameworks. Moreover, the geometry of the parent manifolds is largely understood and supported by robust theoretical foundations. The advantage is that each computation involved in defining parameters of the submanifold can be explained using the interpretations of linear algebra and Riemannian geometry. These explanations and interpretations are transparent and supplement a more precise critique of potential flaws in previous methods and applications. 

The goal is to use a given ensemble of discrete airfoils to infer submanifolds of either $\mathcal{G}(n,2) \times GL_2 \subset \mathbb{R}^{n \times 2}_*$ or $\mathcal{G}(n,2) \times S_{++}^2 \subset \mathbb{R}^{n \times 2}_*$. Over the former, parametrizations through $GL_2$ are often prescribed by the design problem under consideration---e.g., using \eqref{eq:eg_GL2} to enforce structural/regulatory constraints dictating that a specified airfoil thickness at a given location along a wind turbine blade. In this case, scale variations and rotations are prescribed such that a fixed or constrained $M(\bm{\ell})$ is paired with unknown $\bm{t}$. Hence, the geometry of $GL_2$ is largely inconsequential for aerodynamic design given explicit definitions of $M(\bm{\ell})$, and only inferences involving $\mathcal{G}(n,2)$ are required. For $\mathcal{G}(n,2) \times S_{++}^2$, a more flexible representation of airfoil shapes is offered, independent of rotations/reflections. This is most useful for more general representations of airfoil shapes if scale variations independent of rotations are not prescribed---i.e., unknown $(\bm{t}, \bm{\ell})$. This insight provides a blueprint for leveraging different representations of airfoil shapes. That is, if you understand scale variations (e.g., traditional blade design), then compute $M(\bm{\ell})$ explicitly and pair it with inferred statistical properties of $[\tilde{X}](\bm{t})$ from data projected onto $\mathcal{G}(n,2)$ to parametrize shape deformations using~\eqref{eq:sep_shape}. Otherwise, attempt to infer statistical properties for pairs of $[\tilde{X}](\bm{t})$ and $P(\bm{\ell})$  over $\mathcal{G}(n,2) \times S^2_{++}$, and parametrize shape deformations using~\eqref{eq:polar_sep} for more generalized airfoil design.

\subsubsection{Intrinsic Maps}
There are two important intrinsic maps of Riemannian manifolds that are used to parametrize data-driven submanifolds. First, we require the exponential map $\text{Exp}_{p}:T_{p}\mathcal{M} \rightarrow \mathcal{M}$ that parametrizes an initial value problem for geodesic trajectories along the respective manifold $\mathcal{M}$ beginning at $p\in \mathcal{M}$. Specifically, a geodesic curve beginning at $p$ is parametrized over a direction $\bm{v}\in T_p\mathcal{M}$ as $\bm{\gamma}(t;p, \bm{v}) = \text{Exp}_p(t\bm{v})$ such that $\bm{\gamma}(0;p,\bm{v}) = p$. Second, we require the corresponding inverse $\text{Log}_{p}: \mathcal{M} \rightarrow T_{p}\mathcal{M}$ which parametrizes a boundary-value problem given two points connected by a geodesic trajectory over the respective $\mathcal{M}$. Fig.~\ref{fig:intrinsic_maps} shows an example visualization of the $\text{Exp}$ and $\text{Log}$ maps over a $2$-sphere. Compositions involving these maps with a corresponding basis in a fixed tangent space $T_{p}\mathcal{M}$ define \emph{normal coordinates} of the manifold~\cite{lee2006riemannian}. We leverage reduced-dimension subspaces of these normal coordinates to parametrize discrete shape submanifolds from data. 

\subsubsection{Distances}
Distance over a Riemannian manifold $(\mathcal{M}, g)$ is defined as the infimum over a set of lengths---i.e., line integrals of the inner product of curve velocities induced by $g$---connecting two points in the manifold. Per the developments of~\cite{lee2006riemannian}, this definition of distance is consistent with that of zero acceleration curves (or geodesics) over the manifold such that these curves have a length-minimizing property. Consequently, geodesics induce an intuitive notion of shortest-distance curves whose lengths inform a metric over geodesic balls. These geodesic balls are defined as the image of $\text{Exp}$ at a point over an open or closed ball in normal coordinates, as visualized in Fig.~\ref{fig:nml_coords}. When restricted to a geodesic ball over the manifold, these geodesic distances are equivalently represented by the norm over the image of the $\text{Log}$ map~\cite{lee2006riemannian},
\begin{equation} \label{eq:dist}
    d_{\mathcal{M}}:\mathcal{M} \times \mathcal{M} \rightarrow \mathbb{R}_+:(p,z)\mapsto \Vert \text{Log}_p(z) \Vert_g,
\end{equation}
with the norm induced by the metric $g$ at the corresponding point. In our case, the norm in \eqref{eq:dist} is the Frobenius norm, inducing $d_{\mathcal{G}}$ and $d_{S^2_{++}}$ with corresponding $\text{Log}$ for the respective manifolds. In our data-driven setting, we assume data is concentrated within geodesic balls to leverage \eqref{eq:dist} for computing distances, given the desire to compute $\text{Log}$ and $\text{Exp}$ maps that define normal coordinate charts. However, the assumed restriction to geodesic balls is more generally extensible by constructing an atlas of normal coordinate charts using a collection of disjoint tangent spaces.

\subsubsection{Riemannian Statistics}
We next introduce two algorithms that extend the fundamental data analysis technique of principal component analysis (PCA) to Riemannian manifolds as a so-called principal geodesic analysis (PGA)~\cite{pennec1999probabilities,fletcher2003statistics}. The output of PGA will inform a parametrization over the Grassmannian $\mathcal{G}(n,2)$, which defines the higher order perturbations to airfoil shapes. The basic premise of classical PCA is that, given a set of points as data, we can determine directions that maximize sample covariance $G$ of the point set,
\begin{equation}
    \underset{\bm{v} \, \text{s.t.}\, \bm{v}^{\top}\bm{v} = 1}{\text{argmax}} \, \bm{v}^{\top}G\bm{v}.
\end{equation}
Writing the Lagrangian and solving for the stationarity condition in the optimization problem implies $(G - \lambda I)\bm{v} = 0$ for strictly nonnegative $\lambda$. Thus, examining the decreasing ordered pairs $(\lambda_i, \bm{v}_i)$ of eigenvalues and eigenvectors offers a set of directions defining a basis for a reduced-dimension subspace over which the covariance-weighted inner product changes the most, on average. 

We require two functionalities to extend PCA to Riemannian manifolds:  (i) we must \emph{center on the data} by an appropriate notion of an \emph{intrinsic mean} over the manifold; then, (ii) we must \emph{identify important directions} over the manifold that form an orthogonal frame for a \emph{submanifold}. Fig.~\ref{fig:PGA_sphere} offers a useful visual analogue for the case of a sphere. For completeness, we present the two algorithms in detail as Algorithm~\ref{alg:Karcher} and Algorithm~\ref{alg:PGA}. Algorithm~\ref{alg:Karcher} computes the intrinsic (Karcher or Fr\'{e}chet) mean by computing the mean of the data in the tangent space of an iterative approximation restricted over the manifold. Algorithm~\ref{alg:PGA} computes the corresponding PGA directions in the central tangent space using the SVD of the data lifted to the tangent space at the intrinsic mean computed from the previous algorithm. Together, these algorithms provide a framework for parametrizing separable shape tensors. Additional motivation and development for these algorithms can be found in~\cite{fletcher2003statistics, fletcher2004principal, pennec1999probabilities}.

\begin{algorithm}[!h]
\caption{Intrinsic mean over Riemannian manifold}
\label{alg:Karcher}
    \begin{algorithmic}[1]
        \Require Data points $\lbrace p_k \rbrace_{k=1}^N \in \mathcal{M}$ and chosen convergence threshold $\epsilon > 0$
        
        \State Set $p' = p_1$ and initialize $\bm{v}$ such that $\Vert\bm{v}\Vert > \epsilon$
        
        \While{$\Vert \bm{v} \Vert \geq \epsilon$}
            
            \State$\bm{v} = \frac{1}{N} \sum_{k=1}^N\text{Log}_{p'}(p_k)$
            \State$p' = \text{Exp}_{p'}(\bm{v})$
        
        \EndWhile
        
        \noindent
        \Return $p_0 = p' \in \mathcal{M}$
    \end{algorithmic}
\end{algorithm}

\begin{algorithm}
\caption{Principal geodesic analysis directions}
\label{alg:PGA}
    \begin{algorithmic}[1]
    \Require Data points $\lbrace p_k \rbrace_{k=1}^N \in \mathcal{M}$, chosen convergence threshold $\epsilon > 0$, and dimensionality $r$
    
    \State Compute $p_0$ according to Algorithm~\ref{alg:Karcher} with parameter $\epsilon$
    \State Compute matrix
    $$
    \Delta_N = \frac{1}{\sqrt{N-1}}\left( \text{Log}_{p_0}(p_1), \dots, \text{Log}_{p_0}(p_N)\right)
    $$
    \State Compute reduced SVD, up to $r\leq \text{dim}(\mathcal{M})$,
    $$
    \Delta_N = U_r\Sigma_rV_{r}^{\top}
    $$
    \State Compute normal coordinates in the new basis,
    $$
    T = U_r^{\top}\Delta_N = \Sigma_r V_r^{\top}
    $$
    \Return $U_r$ with columns as the $r$ directions in $T_{p_0}\mathcal{M}$ and $T$ as the matrix of \emph{principal normal coordinates}
    \end{algorithmic}
\end{algorithm}

In Algorithm~\ref{alg:Karcher}, the choice of norm $\Vert \cdot \Vert$ over the tangent space in the case of matrix manifolds is taken as the Frobenius norm---induced by the choice of ambient inner product, $\text{tr}(A^{\top}B)$. For Algorithm~\ref{alg:PGA} applied to matrix manifolds, step 2 requires $\text{vec}:\mathbb{R}^{n \times 2} \rightarrow \mathbb{R}^{2n}$, which stacks the columns of the matrices such that the $k$-th columns of $\Delta_N$ are replaced with $\text{vec}(\text{Log}_{p_0}(p_k))$. This does not modify the induced metric implicit to Algorithm~\ref{alg:PGA}---which maximizes an \emph{approximated} notion of a covariance-weighted inner product~\cite{fletcher2003statistics}---since $\text{tr}(A^{\top}B) = \text{vec}(A)^{\top}\text{vec}(B)$. Consequently, the left singular vectors of $\Delta_N$ still correspond to principal directions in $T_{p_0}\mathcal{M}$ that maximize sample covariance in normal coordinates proportional to $\Delta_N^{\top}\Delta_N$. When mapping forward through these important directions, which are defined by the left singular vectors as columns of $U_r$ to parametrize a submanifold, we must reshape once more prior to composition with $\text{Exp}_{p_0}$. Defining $\text{vec}^{-1}:\mathbb{R}^{2n} \rightarrow \mathbb{R}^{n\times 2}$ such that $\text{vec}^{-1}(\text{vec}(A)) = A$, the $r$-dimensional matrix submanifold of interest becomes
\begin{equation}
    \lbrace p \in \mathcal{M}\,:\, p = \text{Exp}_{p_0}(\text{vec}^{-1}(U_r\bm{t})) \rbrace
\end{equation}
for a vector of PGA coefficients $\bm{t} \in \mathbb{R}^r$ and basis $U_r \in \mathbb{R}^{2n \times r}_*$. Once again, despite the composition with reshaping, the choice of ambient inner product remains consistent for $\bm{v} = U_r\bm{t}$, such that $\text{tr}(\text{vec}^{-1}(\bm{v})^{\top}\text{vec}^{-1}(\bm{v})) = \bm{v}^{\top}\bm{v}$.

\subsubsection{Consistent Deformations}
\label{subsec:consistent_deform}
To supplement 3D design for blade or wing shapes, we will require parallel translation to facilitate the mapping of consistent perturbations to different elements of the manifold. Given $p,z \in \mathcal{M}$, parallel translation acts as an isometry parametrizing zero-acceleration transport of tangent vectors over geodesic curves in the manifold, $\tau_{p \parallel z}:T_{p}\mathcal{M} \rightarrow T_{z}\mathcal{M}$. Note that parallel translation is unique and exists over any curve in $\mathcal{M}$. However, if defined over unique geodesic curves within a normal neighborhood, parallel translation has the advantage of not requiring any memory of the curve. That is, we can always reconstruct the original tangent vectors by ``back-tracing'' parallel translation over the choice of unique geodesics that are intrinsic to the manifold. This alleviates the computational burden by taking advantage of tensor parametrizations that are consistent with geodesic trajectories but do not require the otherwise burdensome numerical integration of the dynamics~\cite{arvanitidis2017latent} for the two special cases of Grassmannian and SPD manifolds.

The mapping $\tau_{p \parallel z}$ preserves the notion of direction within a tangent space such that
\begin{equation}
g_p(\bm{u}, \bm{v}) = g_z(\tau_{p \parallel z}(\bm{u}), \tau_{p \parallel z}(\bm{v}))
\end{equation}
for $\bm{u}, \bm{v} \in T_p\mathcal{M}$ and $p,z \in \mathcal{M}$. Consequently, given a basis defined at a particular point, we can map directions in the span of this basis to new tangent spaces along geodesics parametrized by endpoints, i.e., $\tau_{p \parallel z}(\cdot) = \tau_{\parallel}(\cdot; p, z)$. These connected directions have equivalent inner products taken in the central tangent space, constituting the closest analogue of a consistent parameter direction over distinct tangent spaces.

Applying this process to airfoil shapes, we can use Algorithm~\ref{alg:Karcher} to identify the intrinsic mean shape $p_0$ and its central tangent space $T_{p_0}\mathcal{M}$. Next, using Algorithm~\ref{alg:PGA}, we compute a basis in this tangent space describing dominant perturbations about this mean shape. Given a particular deformation to the mean shape $U_r\bm{t} \in T_{p_0}\mathcal{M}$ with coefficients $\bm{t}$, we can then consistently map these deformations to new shapes via $\tau_{\parallel}(U_r\bm{t}; p_0, p)$. Such a procedure preserves the original notion of direction in the parametrization at $T_{p_0}\mathcal{M}$ but assigns the deformation to the shape $p$. The result is our definition of \emph{consistent deformations} of distinct shapes---specifically, consistency in the inner product at the intrinsic mean.

\subsubsection{Geometry of Riemannian Product Manifolds}
Motivated by the separability of our shape parametrizations from \eqref{eq:polar_sep}, we explore the concept of a product manifold. Given $(\mathcal{M}_1, g_1)$ and $(\mathcal{M}_2, g_2)$, we consider the topology of a Riemannian product manifold $(\mathcal{M}_1 \times \mathcal{M}_2, g_1 + g_2)$. Under this construction, we take advantage of the ability to parametrize geodesics in a componentwise manner. Letting $\text{Exp}_i$ denote the exponential map for the manifold $\mathcal{M}_i$, $(\text{Exp}_{1,p}(t\bm{v}), \text{Exp}_{2,z}(t\bm{u}))$ is the corresponding geodesic over $(\mathcal{M}_1 \times \mathcal{M}_2, g_1 + g_2)$ for $t\in \mathbb{R}$. This allows us to independently formulate the necessary intrinsic maps over distinct manifolds and combine these computations in a componentwise manner to build shapes using our separable representations. Given this convenience, the general submanifold of discrete shapes assumes the defined product manifold topology over $\mathcal{G}(n,2) \times S^2_{++}$ with the canonical ambient metric and affine-invariant metric, respectively.

\subsubsection{Geometry of Orthogonal Matrices: $\mathcal{G}(n,2)$}\label{subsub:grassmann}

Following the developments of~\cite{edelman1998geometry, gallivan2003efficient, bendokat2020grassmann}, we present algorithmic routines for the computation of the $\text{Exp}$ and $\text{Log}$ maps over $\mathcal{G}(n,2)$ for completeness.\footnote{We assume the Riemannian metric $\text{tr}(A^{\top}B)$ inherited from embedding space~\cite{absil2008optimization}.} 

\begin{algorithm}
\caption{Thm. 2.3~\cite{edelman1998geometry}, Grassmann exponential with Stiefel representatives}
\label{alg:Gr_Exp}
    \begin{algorithmic}[1]
        \Require Representative matrix $\tilde{X}\in \mathbb{R}_*^{n\times 2}$ with orthonormal columns and direction $\Delta \in T_{[\tilde{X}]}\mathcal{G}(n,2) \subset \mathbb{R}_*^{n\times 2}$
        
        \State Compute thin SVD, 
        $$
        \Delta = U_{\Delta}\Sigma_{\Delta}V^{\top}_{\Delta}
        $$
        \State Compute, with $\cos(\cdot)$ and $\sin(\cdot)$ acting only on diagonal entries,
        $$
        \tilde{Y} = \tilde{X}V_{\Delta}\cos (\Sigma_{\Delta})V_{\Delta}^{\top} + U_{\Delta}\sin (\Sigma_{\Delta})V_{\Delta}^{\top}
        $$
        \Return $[\tilde{Y}] = \text{Exp}_{[\tilde{X}]}(\Delta) \in \mathcal{G}(n,2)$
    \end{algorithmic}
    
\end{algorithm}

First, we present the exponential map as Algorithm~\ref{alg:Gr_Exp}. Leveraging this algorithm, we can compute $\text{Exp}_{[\tilde{X}]}(\Delta)$ to take a unit step in the direction $\Delta$ from the equivalence class $[\tilde{X}]$. As a reparametrization, we can arbitrarily scale distances along this geodesic via the one-parameter subgroup $\gamma(t; [\tilde{X}], \Delta) = \text{Exp}_{[\tilde{X}]}(t\Delta)$ for all $t \in \mathbb{R}$, identifying the base point as $\gamma(0; [\tilde{X}], \Delta) = [\tilde{X}]$. For general $\mathcal{G}(n,q)$, this algorithm has computational complexity $O(nq^2)$ by virtue of a generalized SVD as the only iterative procedure~\cite{gallivan2003efficient}. In our fixed ambient spatial dimension $q=2$, we anticipate linear scaling of the computational burden with increasing $n$~\cite{gallivan2003efficient, zimmermann2019manifold, bendokat2020grassmann}. 

Next, we present Algorithm~\ref{alg:Gr_Log} for computing the inverse exponential map~\cite{gallivan2003efficient,absil2004riemannian,zimmermann2019manifold, bendokat2020grassmann}. In general, this algorithm has computational complexity $O(nq^2)$, which again implies linear complexity with landmark refinements for fixed ambient spatial dimension $q=2$. However, treating this map simply as a matrix transformation is subject to $\text{Exp}_{[\tilde{X}]}(\text{Log}_{[\tilde{X}]}(\tilde{Y})) \neq \tilde{Y}$ with mismatch up to rotation~\cite{zimmermann2019manifold, bendokat2020grassmann}. The mismatched matrix still corresponds to the same $[\tilde{Y}] = [\hat{Y}]$ such that $\tilde{Y} \sim \hat{Y}$, but it lacks the desirable property that $\text{Exp}_{[\tilde{X}]} \circ \text{Log}_{[\tilde{X}]}$ returns the same element of the equivalence class up to reflections. This computational inconvenience is corrected by Procrustes matching~\cite{zimmermann2019manifold, bendokat2020grassmann}. This concept comes into play for the purposes of 3D blade and wing cross section interpolation, where a sequence of representative matrices of the Grassmannian intended for interpolation can be mismatched up to rotation, requiring Procrustes analysis to align. This is discussed further in Section~\ref{subsec:Grassmann_blade_interp}. Despite this correction, there remains a possibility that the shape is reflected in the composition $\text{Exp}_{[\tilde{X}]} \circ \text{Log}_{[\tilde{X}]}$. Additional checks or constraints may be required for detecting such a condition, depending on the data.

\begin{algorithm}
\caption{Alg. 5.3~\cite{bendokat2020grassmann}, Alg. 10~\cite{zimmermann2019manifold}, Grassmann logarithm with Stiefel representatives}
\label{alg:Gr_Log}
\begin{algorithmic}[1]
    \Require Representative matrices $\tilde{X}, \tilde{Y} \in \mathbb{R}^{n\times 2}_*$ with orthonormal columns
    \State Take the matrix product, $\tilde{Q} = \tilde{X}^{\top}\tilde{Y}$
    \State Define orthogonal projection (to normal space),
    $$
    \pi_n^{\perp} =(I_n - \tilde{X}\tilde{X}^{\top})
    $$
    \State Compute thin SVD, 
    $$
    \pi^{\perp}_n(\tilde{Y}\tilde{Q}^{-1}) = U_{\Delta}\Sigma_{\Delta}V_{\Delta}^{\top}
    $$
    \State Compute, with $\arctan(\cdot)$ acting only on diagonal entries, 
    $$
    \Delta = U_{\Delta}\arctan(\Sigma_{\Delta})V_{\Delta}^{\top}
    $$
    \Return $\Delta = \text{Log}_{[\tilde{X}]}([\tilde{Y}]) \in T_{[\tilde{X}]}\mathcal{G}(n,2)$
\end{algorithmic}
\end{algorithm}

A modified version of Algorithm~\ref{alg:Gr_Log} such that $\text{Exp}_{[\tilde{X}]}(\text{Log}_{[\tilde{X}]}(\tilde{Y})) = \tilde{Y}$ up to reflections can be found in~\cite{zimmermann2019manifold,bendokat2020grassmann} but is omitted from this presentation for brevity. The motivation here is to highlight the relatively inexpensive computations of these intrinsic mappings---linear growth in computational complexity for refined shapes in fixed ambient spatial dimension---for the purposes of informing a statistical analysis and separable representation of shapes. The modified version of Algorithm~\ref{alg:Gr_Log} requires two SVD computations but remedies the representative rotational mismatch and, as an added benefit, avoids the calculation of $\tilde{Q}^{-1}$~\cite{zimmermann2019manifold}.

\begin{algorithm}
\caption{Thm. 2.4~\cite{edelman1998geometry}, Grassmann parallel transport along geodesic with Stiefel representatives}
\label{alg:Gr_parallel}
    \begin{algorithmic}[1]
        \Require Representative matrix $\tilde{X}\in \mathbb{R}_*^{n\times 2}$ with orthonormal columns and directions $\Delta, \Gamma \in T_{[\tilde{X}]}\mathcal{G}(n,2) \subset \mathbb{R}_*^{n\times 2}$, and scalar $t \in \mathbb{R}$
        
        \State Compute thin SVD, 
        $$
        \Delta = U_{\Delta}\Sigma_{\Delta}V^{\top}_{\Delta}
        $$
        
        \State Compute parallel transport operator, with $\cos(\cdot)$ and $\sin(\cdot)$ acting only on diagonal entries,
        
        \begin{multline*}
            \tau_{\parallel}(t;[\tilde{X}], \Delta)= \left(\tilde{X}V_{\Delta}\,\, U_{\Delta} \right)\left( \begin{matrix}
            -\sin (t\Sigma_{\Delta})\\
            \cos (t\Sigma_{\Delta})
            \end{matrix}\right)
            U_{\Delta}^{\top} \\
            + I_n - U_{\Delta}U_{\Delta}^{\top}
        \end{multline*}
        
        \noindent
        \Return $\Gamma(t;[\tilde{X}], \Delta) = \tau_{\parallel}(t;[\tilde{X}], \Delta)\Gamma$
    \end{algorithmic}
    
\end{algorithm}

Lastly, we present the algorithm for parallel translation~\cite{edelman1998geometry}. Using Algorithm~\ref{alg:Gr_parallel}, $\Gamma(t; [\tilde{X}], \Delta)  \in T_{\gamma(t;[\tilde{X}], \Delta )}\mathcal{G}(n,2)$ is the parallel translation of $\Gamma$ along the geodesic a distance scaled by $t$ emanating from $[\tilde{X}]$ in the direction $\Delta$. Again, a comparable computational complexity is achieved by only requiring the thin SVD of the direction $\Delta$---defining the geodesic $\gamma(t;[\tilde{X}], \Delta))$---over which $\Gamma$ is parallel translated by the $n\times n$ matrix $\tau_{\parallel}(t;[\tilde{X}], \Delta)$.

Notice that Algorithm~\ref{alg:Gr_parallel} is parametrized by a ``point'' and a ``direction'' as opposed to two endpoints, per the development of consistent deformations, $\tau_{p\parallel z}(\cdot) = \tau_{\parallel}(\cdot; p,z)$. Both interpretations are valid (within a geodesic ball) because the direction and magnitude parametrizing the geodesic from $p$ to $z$ (generally) is implied by $\text{Log}_{p}(z) \in T_{p}\mathcal{M}$, while the endpoints in Algorithm~\ref{alg:Gr_parallel} are implied by $\gamma(0; [\tilde{X}], \Delta)$ and $\gamma(t; [\tilde{X}], \Delta)$ using Algorithm~\ref{alg:Gr_Exp}. In our implementations, the endpoint parametrization is most appropriate, thus modifying the Grassmannian parallel translation of $\Gamma \in T_{[\tilde{X}]}\mathcal{G}(n,2)$ to $T_{[\tilde{Y}]}\mathcal{G}(n,2)$, as $\tau_{\parallel}(\Gamma;[\tilde{X}],[\tilde{Y}])  = \tau_{\parallel}(1;[\tilde{X}], \text{Log}_{[\tilde{X}]}([\tilde{Y}]))\Gamma$ by composing Algorithm~\ref{alg:Gr_parallel} with Algorithm~\ref{alg:Gr_Log}.

\subsubsection{Geometry of the Symmetric Space: $S^2_{++}$}
Next, we review the necessary computations over the convex cone of $2\times 2$ SPD matrices, $S^2_{++}$. We follow the developments of \cite{pennec2006riemannian, bonnabel2010riemannian, fletcher2004principal} for algorithms to compute $\text{Exp}$ and $\text{Log}$ maps, while \cite{sra2015conic} additionally expound on the computation of parallel transport. More recent implementations and interpretations are facilitated by \cite{yair2019parallel, zimmermann2019manifold}.\footnote{Consistent with the referenced developments, we assume the \emph{affine-invariant metric} for all computations over $S^2_{++}$.} For brevity, we reiterate the presentation of \cite{zimmermann2019manifold} and note that a consistent yet more systematic implementation is discussed in \cite{fletcher2004principal} for the purposes of computing $P^{1/2}$ and $P^{-1/2}$. 

\begin{algorithm}
\caption{\cite{zimmermann2019manifold, fletcher2004principal}, SPD exponential}
\label{alg:SPD_Exp}
    \begin{algorithmic}[1]
        \Require Matrix $P\in S^2_{++}$ and direction $S \in T_{P}S^2_{++} \cong \text{Sym}(2)$
        
        \State $D = P^{1/2}\exp(P^{-1/2}SP^{-1/2})P^{1/2}$
        
        \noindent
        \Return $D = \text{Exp}_{P}(S) \in S^2_{++}$
    \end{algorithmic}
    
\end{algorithm}

The exponential map over $S^2_{++}$ is stated as Algorithm~\ref{alg:SPD_Exp}. This algorithm induces the one-parameter subgroup for moving along geodesics emanating from $P$ in the direction $S$ as a $2\times 2$ symmetric matrix, $\gamma(t; P, S) = \text{Exp}_{P}(tS)$. Note that ``$\exp$'' represents the matrix exponential (distinct from ``$\text{Exp}$,'' given by the algorithm for a particular choice of Riemannian metric). 

\begin{algorithm}
\caption{\cite{zimmermann2019manifold, fletcher2004principal}, SPD logarithm}
\label{alg:SPD_Log}
    \begin{algorithmic}[1]
        \Require Matrices $P, D \in S^2_{++}$
        
        \State $S = P^{1/2}\log(P^{-1/2}DP^{-1/2})P^{1/2}$
        
        \noindent
        \Return $S = \text{Log}_{P}(D) \in T_{P}S^2_{++} \cong \text{Sym}(2)$
    \end{algorithmic}
    
\end{algorithm}

The inverse exponential map is stated as Algorithm~\ref{alg:SPD_Log}. Note that in this algorithm, ``$\log$'' represents the matrix logarithm (distinct from ``$\text{Log}$'').

\begin{algorithm}
\caption{\cite{sra2015conic,yair2019parallel}, SPD parallel translation along geodesic}
\label{alg:SPD_parallel}
    \begin{algorithmic}[1]
        \Require Matrices $P, D \in S^2_{++}$ and direction $S \in T_{P}S^2_{++} \cong \text{Sym}(2)$
        
        \State Compute $E = (DP^{-1})^{1/2}$
        
        \noindent
        \Return $ESE^{\top} = \tau_{\parallel}(S;P, D) \in T_{D}S^2_{++}$
    \end{algorithmic}
    
\end{algorithm}

Finally, parallel translation is given by Algorithm~\ref{alg:SPD_parallel}. In this algorithm, parallel translation is parametrized by the two endpoints of the geodesic curve, $\tau_{\parallel}(S; P, D) = \tau_{P\parallel D}(S)$, consistent with the generalizations described in Section~\ref{subsec:consistent_deform}.

These routines, combined with the Grassmannian routines, offer a complete picture of the necessary computations for learning a manifold of discrete shapes for 2D and 3D blade design and interpolation. The improved notion of distances between shapes informs more favorable distributions for numerical studies and supplements improved shape representations by regularizing deformations---i.e., constraining to data-driven submanifolds. Moreover, this approach constitutes a more principled perspective for learning a submanifold of discrete shapes from data with reduced computational costs compared to alternative ML-based methods.

\subsection{Grassmannian Blade Interpolation}
\label{subsec:Grassmann_blade_interp}
We discuss the procedure for applying the framework of separable shape tensors to interpolate a sequence of 2D shapes into a 3D blade/wing. We also discuss the implications of a \emph{Procrustes clustering} approach to select best representative matrices from equivalence classes for interpolation.


\begin{figure*}[t]
	\centering
	\includegraphics[width=\textwidth]{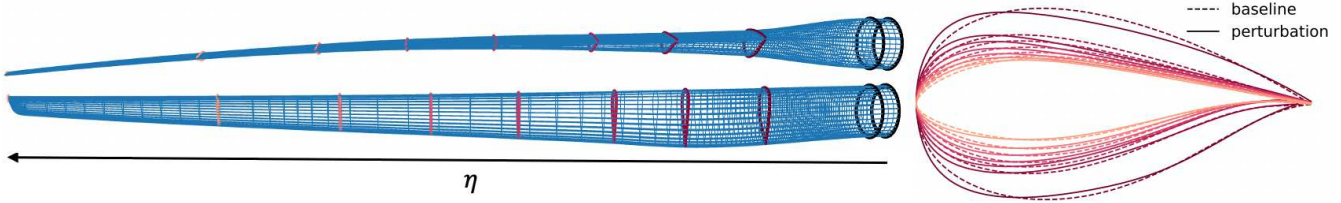}
	\caption{(left) Example of an IEA 15-MW blade wire-frame obtained from interpolation of the solid-color cross sections (right) omitting the unperturbed circles. Note that consistent perturbations to the shape (right) are applied to all of the baseline airfoils in the blade. Visually, we observe each of the distinct airfoils sampled from distinct classes deformed in a markedly similar manner. Also, note the curved span axis (bending along $\eta$), which is easily accounted for by including planar shapes in $\mathbb{R}^3$, then rotating the shapes into a plane with normal tangent to the span axis.}
	\label{fig:interp_blade}
\end{figure*}
\begin{figure*}[t]
	\centering
	\begin{subfigure}{0.65\textwidth}
         \includegraphics[width=0.9\textwidth]{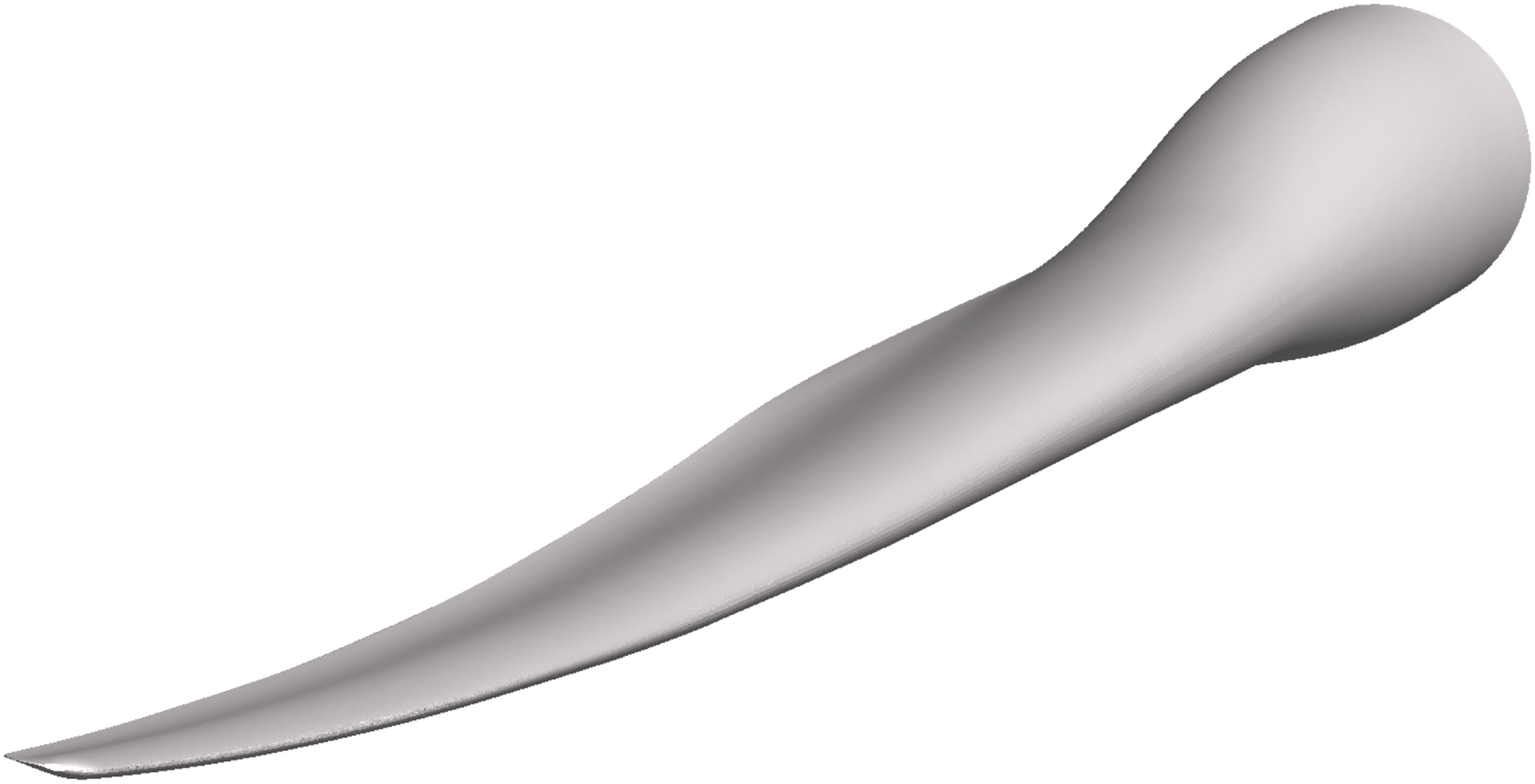}
     \end{subfigure}
     \hfill
     \begin{subfigure}{0.3\textwidth}
         \includegraphics[width=1\textwidth]{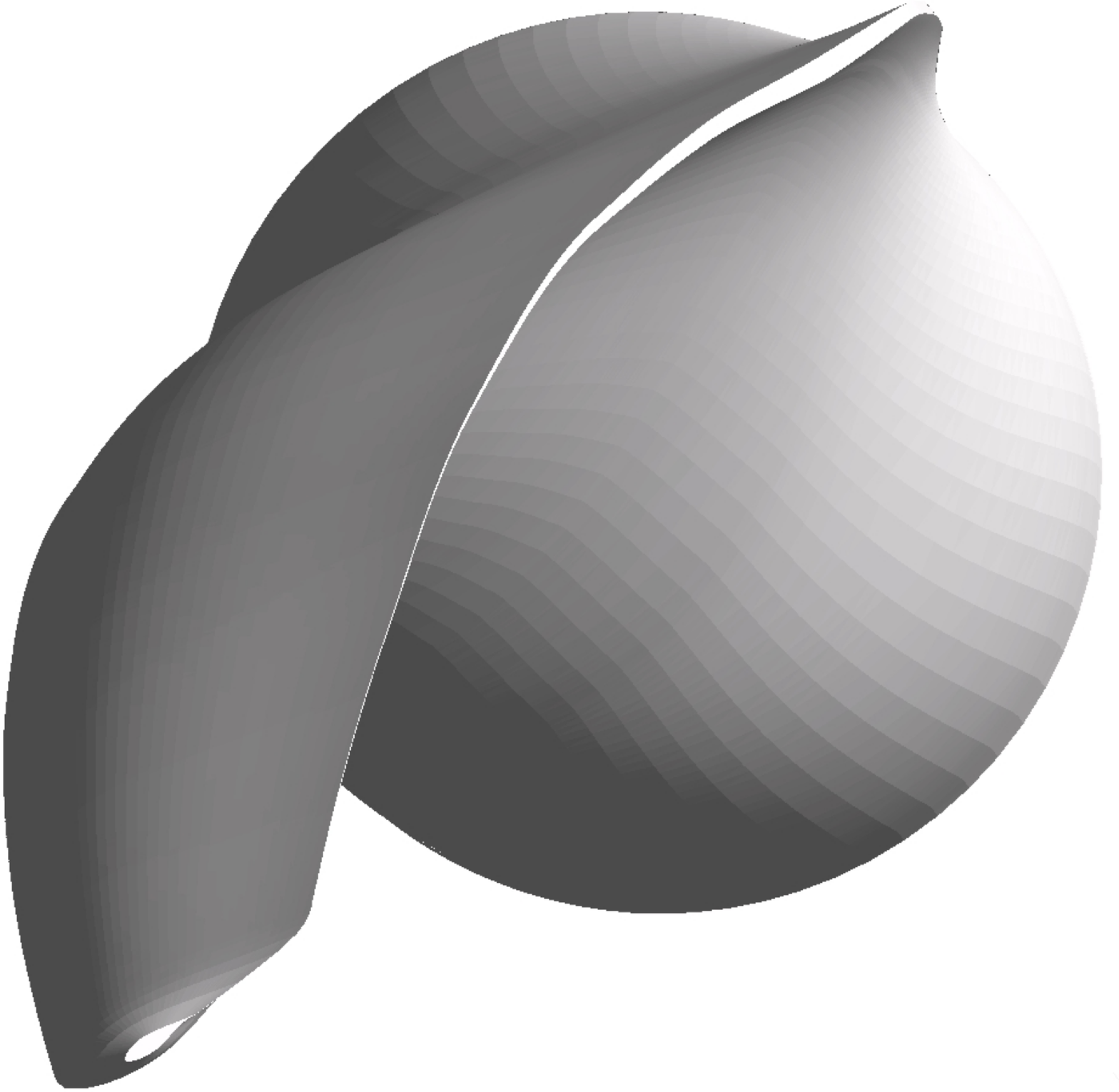}
     \end{subfigure}
     \hfill
	\caption{Structured surface mesh of $10$ nominal cross sections interpolated to $100$ refined cross sections, defining a ``through curves'' surface representation from B-splines of discrete shapes with $n=401$. The mesh was generated using Gmsh~\cite{geuzaine2009gmsh} with the Grassmannian interpolation wire-frame provided as input. The input wire-frame to the surface meshing routine is consistent with the wire-frame depicted in Fig.~\ref{fig:interp_blade}.}
	\label{fig:mesh_blade}
\end{figure*}

The separable shape tensor framework for airfoil representation has the added benefit of enabling the design of 3D wings and blades. In the context of wind energy, blade shapes are defined by a limited number of landmark airfoils located at different blade-span positions, as well as by profiles of twist, chordal scaling, and bending. Defining the full blade shape given the relatively small number of defining airfoil shapes along the blade is a nontrivial problem. Simple interpolation techniques result in undesirable blade features, such as kinks or dimples, in the regions between airfoils. Currently, airfoils must be designed as collections or families that will interpolate smoothly to construct the blade. The goal here is to define an interpolation of these shapes---independent of the prescribed affine deformations---that results in physically relevant blade definitions. In addition to interpolating between designed airfoil shapes, we may seek a separable representation from measured blades and subsequently infer a smoothly varying set of affine deformations over discrete blade-span positions corresponding to twist, scaling, and bending profiles of the blade. 

\subsubsection{Interpolation Procedure}
Given a sequence of matrices $(X_k) \in \mathbb{R}_*^{n\times2}$ for $k=1,\dots,N$ that represent landmark airfoil shapes in a wing or turbine blade, we induce the corresponding sequence of equivalence classes $([\tilde{X}_k]) \in \mathcal{G}(n,2)$ located at blade-span positions $\eta_k \in \mathcal{S} \subset \mathbb{R}$ from root to tip. An example of such landmark shapes is depicted by the red-yellow curves in Fig.~\ref{fig:interp_blade}. We define a piecewise geodesic path over the Grassmannian to interpolate representative airfoil shapes from $([\tilde{X}_k])$. This results in a continuous representation of the 3D blade shape using piecewise geodesic paths over ordered blade-span positions $\eta_k$ along a nonlinear representative manifold of shapes. In practice, we use piecewise geodesic interpolation $\bm{\gamma}(\tilde{t}; [\tilde{X}_k], \text{Log}_{[\tilde{X}_k]}([\tilde{X}_{k+1}]))$ via Algorithms~\ref{alg:Gr_Exp} and \ref{alg:Gr_Log} for all $\tilde{t} \in [0,1]$ and $k=1,\dots,N-1$ enumerating the sequence of geodesics interpolating between representative shapes over the Grassmannian. This procedure is also described as Algorithm 2 of \cite{zimmermann2019manifold} and is summarized below.

To implement this approach, we must first reconcile differences in length scales over the piecewise geodesic Grassmannian curve and the spanwise physical distances between the airfoils within the blade. We define a monotonic reparametrization to consider a mapping from the physically relevant blade-span position $\eta \in \mathcal{S}$ to the corresponding cumulative distance over the Grassmannian, $\varphi:\eta \mapsto t$. In practice, this mapping can be built from a PCHIP of $\lbrace( \eta_k, t_k)\rbrace$ as a monotonically increasing function such that
\begin{equation}
    t_k = \begin{cases}
    0, \,\, k=1\\
    \sum_{i=1}^{k-1}d_{\mathcal{G}}([\tilde{X}_{i+1}], [\tilde{X}_{i}]), \,\, k=2,\dots,N.
    \end{cases}
\end{equation}
Then, within any piecewise interval $[t_k,t_{k+1}] = [\varphi(\eta_k), \varphi(\eta_{k+1})]$, we scale $[t_k,t_{k+1}]$ to $[0,1]$ to build interpolated shapes over the subinterval, informing $[\tilde{X}](\varphi(\eta))$ for all $\eta \in [\eta_k, \eta_{k+1}]$. As a three-step procedure, given any $\eta\in \mathcal{S}$: (i) convert to the appropriate cumulative Grassmannian distance $t =\varphi(\eta)$ and identify the corresponding subinterval, $t \in [t_k,t_{k+1}]$; (ii) scale to a normalized coordinate $\tilde{t} = (\varphi(\eta) - t_k)/(t_{k+1} - t_k)$; and (iii) compute $\text{Exp}_{[\tilde{X}_k]} (\tilde{t}\text{Log}_{[\tilde{X}_k]}([\tilde{X}_{k+1}])$ using the composition of Algorithms~\ref{alg:Gr_Exp} and \ref{alg:Gr_Log}. Finally, to map the interpolated shapes over the Grassmannian back to physically relevant scales, we apply appropriate affine deformations using six regularized splines of data or explicit parametrizations $M(\eta)$ and $\bm{b}(\eta)$,
\begin{equation} \label{eq:blade}
	X(\eta) = (\tilde{X} \circ \varphi)(\eta)M(\eta)+1_{n,2}\text{diag}(\bm{b}(\eta)).
\end{equation}
In Section~\ref{subsub:procrustes}, we discuss the implications of inferring $M(\eta)$ and $\bm{b}(\eta)$ from measured data $(X_k)$ using splines.

As an example computation, on a laptop (2.4 GHz 8-Core Intel Core i9 macOS Catalina Memory: 32 GB 2667 MHz DDR4), the interpolation routine with reparametrized shapes according to $n=401$, with $N=10$ nominal cross sections as input and a refinement of $100$ cross sections defining the wire-frame, took approximately $0.04$ seconds, on average. The corresponding blade is shown in Fig.~\ref{fig:interp_blade}, with a resulting structured surface mesh shown in Fig.~\ref{fig:mesh_blade}. Varying refinements up to $1,000$ new cross sections defining the wire-frame took $0.12$ seconds, and refinements up to $100,000$ new cross sections took $14.9$ seconds, on average (all else fixed). For comparison, it often takes longer to read the $10$ nominal cross sections into memory than it does to run the interpolation routine with refinements between $100$ and $1,000$ cross sections. Code and examples are available at \cite{doecode_73484}.

\subsubsection{Procrustes Clustering} \label{subsub:procrustes}
If the affine scale variations are implicitly encoded in the original data $(X_k)$---i.e., the sequence of discrete airfoils has already been appropriately scaled to size and orientation and do not constitute shapes with fixed orientation and unit-chord---then computations of the discrete centers of mass $\bm{b}_k = 1/nX_k^{\top}1_{n,1}$ and $( \tilde{M}_{k})$, using~\eqref{eq:LA_M} for $(X_k)$, can be utilized. With data $\lbrace(\eta_k,\tilde{M}_{k},\bm{b}_k)\rbrace$, we construct six entrywise splines over strictly increasing $(\eta_k)$, defining $M(\eta) = \tilde{M}(\eta)$ and $\bm{b}(\eta)$ in~\eqref{eq:blade} with appropriate endpoint conditions. However, large uncontrolled variations in $(\tilde{M}_{k})$, namely from rotation, may be problematic in the construction of the four corresponding entrywise splines. 

Recalling the concerns with Algorithm~\ref{alg:Gr_Log} about a rotational mismatch between representative shapes from equivalence classes, we perform a \emph{Procrustes clustering} along the reversed order of representative shapes---i.e., for $k=N,\dots,2$, we solve
\begin{equation} \label{eq:Procrustes}
    R_{k-1} = \underset{R \in \mathcal{O}(2)}{\text{argmin}} \,\, \Vert \tilde{X}_{k} - \tilde{X}_{k-1}R \Vert_F
\end{equation}
then apply the computed rotation to the LA-standardized shape $\tilde{X}_{k-1}R_{k-1}$ and reassign scale variations to $\tilde{M}_{k-1}R_{k-1}$ so that the reversed\footnote{Our intuition is that wind turbine shapes are typically more similar from tip to hub, and the hub shape is often circular and thus invariant under rotation. Thus, reversing the order may be desirable but is seemingly unnecessary.} sequence of shapes are best matched (sequentially) for interpolation. 

This is an important caveat when inverting the shapes in~\eqref{eq:blade} back to the physically relevant scales for subsequent affine deformations inferred from data. As a procedural interpretation, from the blade tip shape $\tilde{X}_{N}$ to the blade hub shape $\tilde{X}_1$, we sequentially match the representative LA-standardized shapes via Procrustes analysis~\cite{gower1975generalized} using~\eqref{eq:Procrustes}. This offers rotations that can be applied to representative LA-standardized airfoils for matching---but does not modify the underlying piecewise geodesic interpolation along the Grassmannian, which is independent of these rotations. Consequently, we cluster the sequence of representative shapes $\tilde{X}_k$ by optimal rotations in each $[\tilde{X}_k]$ to ensure they are best oriented from tip to hub and to mitigate concerns about large variations in entrywise splines of $\tilde{M}(\eta)$.

\section{Data-Driven Representations}
We apply the detailed developments of the Riemannian interpretations to learn the manifold coordinates $(\bm{t},\bm{\ell})$ that define a generative airfoil model from data in seconds. We then demonstrate the generation of novel 3D blade shapes by consistently deforming landmark airfoils with chosen parameter dimensionality independent from the number of interpolated 2D cross sections.

To explore separable shape representations, we use a data set containing $1,000$ perturbations to $16$ baseline airfoils for a total of $16,000$ shapes from the NREL 5-MW, DTU (Technical University of Denmark) 10-MW, and IEA 15-MW reference wind turbines~\cite{JonkmanBMS:2009,bak2013description,IEA15MW_ORWT}. The baseline airfoils are defined by $18$ nominal CST coefficients with the trailing edge thickness coefficients set to zero. We then perturb these $18$ coefficients by up to $\pm 20\%$ of their nominal value to create an ensemble of random airfoils. 

The left plot of Fig.~\ref{fig:scatter} shows a marginal 2D slice through the 18-dimensional space of CST coefficients that defines the collection of shapes under consideration. Note that across the $16$ baseline shapes, the groups of perturbations to nominal CST coefficients create a complex, highly disjoint design domain given the variety of airfoil classes. This can significantly impact the performance of various AI/ML algorithms in analyzing airfoils across this domain. We demonstrate how the proposed separable representation addresses these issues with the CST representation.

\subsection{Principal Geodesic Deformations}
We describe the process for applying the framework of separable shape tensors to learning coordinates ${(\bm{t},\bm{\ell}) \in \mathcal{T} \times \mathcal{P}}$ for aerodynamic design and generation of new shapes. We explicitly parametrize two submanifolds used to define the shape product manifold. Additionally, interpolation is revisited in this context for the construction of 3D wings/blades from sequences of 2D shapes. Numerical examples demonstrate an improvement in shape generation along random continuous sweeps over the inferred Grassmannian parameter domain.

\begin{figure}
	\centering
	\includegraphics[width=0.95\linewidth]{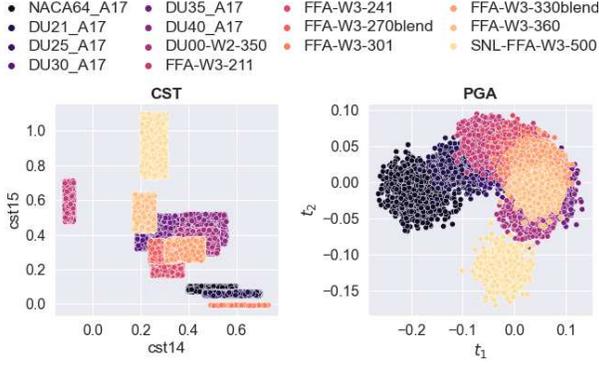}
	\caption{Comparison of the airfoil data over two of the 18 total CST parameters (left) and two of the four total normal coordinates (right), with colors indicating different named classes of designed airfoils. Notice that the transformed normal coordinates resemble a Gaussian mixture.}
	\label{fig:scatter}
\end{figure}
\begin{figure*}
	\centering
	\includegraphics[width=\linewidth]{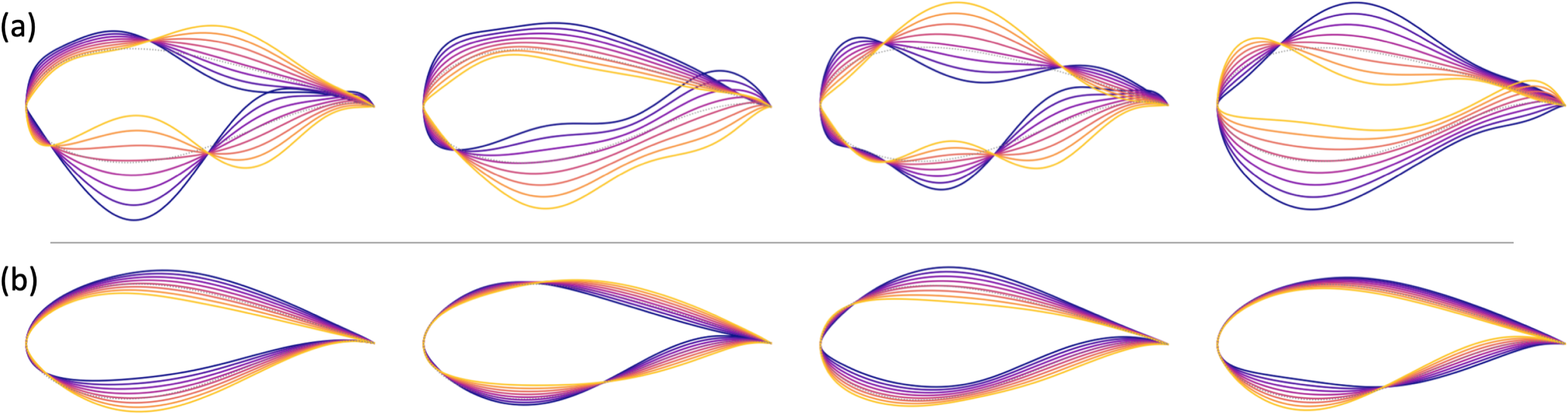}
	\caption{A series of random corner-to-corner sweeps across a hypercube containing all points partially visualized in Fig.~\ref{fig:scatter} through (a) CST parameter space and (b) the Grassmannian submanifold $\mathcal{A}_4$. The sweeps represent geodesic trajectories from yellow to blue shapes passing through the origin of the respective coordinate systems---CST parameter space is assumed to be flat.}
	\label{fig:corner_sweeps}
\end{figure*}

We infer nonparametric data-driven airfoil deformations by applying PGA (Algorithm~\ref{alg:PGA}), using Algorithms~\ref{alg:Gr_Exp} and \ref{alg:Gr_Log} and/or Algorithms~\ref{alg:SPD_Exp} and \ref{alg:SPD_Log} to compute the required intrinsic maps for the respective matrix manifolds. This informs an approximated \emph{central tangent space} at a Karcher mean, denoted $T_{[\tilde{X}_0]}\mathcal{G}(n,2)$ or $T_{[P_0]}S^2_{++}$, computed with Algorithm~\ref{alg:Karcher} using the corresponding intrinsic maps for either manifold of interest. The subsequent image of $\text{Exp}$, given at respective Karcher means over subspaces $\mathcal{T} \subseteq T_{[\tilde{X}_0]}\mathcal{G}(n,2)$ or $\mathcal{P} \subseteq T_{P_0}S^2_{++}$, define \emph{planar sections}~\cite{lee2006riemannian} of the manifolds as submanifolds, $\text{Exp}_{[\tilde{X}_0]}(\mathcal{T}) \subseteq \mathcal{G}(n,2)$ or $\text{Exp}_{P_0}(\mathcal{P}) \subseteq S^2_{++}$.  In this way, PGA and the product manifold construction constitute a manifold learning procedure for computing important submanifolds that represent a design space of physically relevant airfoil shapes inferred from provided data. 

Based on the strength of the decay in eigenvalues over $T_{[\tilde{X}_0]}\mathcal{G}(n,2)$, we take the first $r$ eigenvectors as a reduced basis for Grassmannian PGA deformations. For PGA over the SPD matrix manifold, we retain all three dimensions of $S^2_{++}$. To compute coordinates $\bm{t}_k \in \mathcal{T}$ or $\bm{\ell}_k \in \mathcal{P}$ corresponding to each $X_k$ in a new basis returned by Algorithm~\ref{alg:PGA}, we map LA-standardized airfoils to normal coordinates of $T_{[\tilde{X}_0]}\mathcal{G}(n,2)$ or $T_{P_0}S^2_{++}$ via inner products with the computed basis. In particular, in Algorithm~\ref{alg:PGA}, we compute 
\begin{equation}
    \bm{t}_k = T_r^{\top}\text{vec}(\text{Log}_{[\tilde{X}_0]}([\tilde{X}_k]))
\end{equation}
or 
\begin{equation}
    \bm{\ell}_k = L_3^{\top}\text{vec}(\text{Log}_{P_0}(P_k)) 
\end{equation}
with $T_r$ and $L_3$ as the approximated bases returned by Algorithm~\ref{alg:PGA} applied to either matrix manifold of interest.\footnote{Note that the ``$\text{vec}$'' operation applied to symmetric matrices in $T_{P_0}S^2_{++}$ only ``vectorizes" the three unique entries. In a similar manner, $\text{vec}^{-1}$ in this sense would duplicate the corresponding symmetric entry when mapping to the representative matrix, such that $\text{vec}^{-1}(\text{vec}(P)) = P$.} Notice that the coordinates $\bm{t}_k$ and $\bm{\ell}_k$ correspond to the columns of the principal normal coordinate matrix $T$ returned by Algorithm~\ref{alg:PGA}. This implies definitions of parameter spaces as $\mathcal{T} \subseteq T_{[\tilde{X}_0]}\mathcal{A}_r$ and $\mathcal{P} \subseteq T_{P_0}S^2_{++}$. 

As an example, on a laptop (2.4 GHz 8-Core Intel Core i9 macOS Catalina Memory: 32 GB 2667 MHz DDR4), we randomly subsampled $N=10,000$ shapes from a database of $35,035$ total shapes. Ten sets of random draws with $N=10,000$ were then fed to Algorithm~\ref{alg:Karcher} and Algorithm~\ref{alg:PGA} which were applied over the Grassmannian for $n=401$ and $\epsilon=1e-8$. Algorithm~\ref{alg:Karcher} ran in $14.3$ seconds, on average, over the ten draws, while Algorithm~\ref{alg:PGA}---including the computational time of Algorithm~\ref{alg:Karcher}---ran in $18.2$ seconds (an additional $3.9$ seconds) on average. Again, code and examples are available at \cite{doecode_73484}.

Having defined reduced-dimension normal coordinates $\bm{t}_k \in \mathcal{T}$ with paired $\bm{\ell}_k \in \mathcal{P}$, we restrict $\mathcal{T}$ and $\mathcal{P}$ as compact sets over the respective tangent spaces, which contain PGA coordinates described by an appropriate distribution---e.g., uniform over a ball containing the data $\lbrace \bm{t}_k\rbrace$ and/or $\lbrace(\bm{t}_k,\bm{\ell}_k) \rbrace$. 
Then, the set of all linear combinations of the principal components, $T_r\bm{t}$ for all $\bm{t} \in \mathcal{T}$ and $L_3\bm{\ell}$ for all $\bm{\ell} \in \mathcal{P}$, define an $(r+3)$-dimensional domain over $\mathcal{T} \times \mathcal{P} \subset T_{[\tilde{X}_0]}\mathcal{G}(n,2) \times T_{P_0}S^2_{++}$. This parametrizes a section of the Grassmannian ($r$-submanifold) for all $\bm{t} \in \mathcal{T} \subset \mathbb{R}^r$ and $\Delta(\bm{t}) = \text{vec}^{-1}(T_r\bm{t})$,
\begin{equation} \label{eq:Gr_submanifold}
	\mathcal{A}_r = \left\lbrace [\tilde{X}] \in \mathcal{G}(n,2) \,:\, [\tilde{X}] = \text{Exp}_{[\tilde{X}_0]}(\Delta(\bm{t}))\right\rbrace,
\end{equation}
and the SPD matrix manifold for all $\bm{\ell} \in \mathcal{P}$,
\begin{equation} \label{eq:SPD_submanifold}
    \mathcal{S}_3 = \left\lbrace P \in S^2_{++} \,:\, P = \text{Exp}_{P_0}(\text{vec}^{-1}(L_3\bm{\ell}))\right\rbrace.
\end{equation}
The resulting parametrized Riemannian submanifold ${\mathcal{A}_r \times \mathcal{S}_3}$ becomes the product manifold of interest for generating rotation/reflection-invariant airfoils $\tilde{X}(\bm{t})P(\bm{\ell})$ according to \eqref{eq:polar_sep}. Moreover, when composed with subsequent rotations and spanwise reparametrizations, the data-driven representation \eqref{eq:polar_sep} can also be used to interpolate any blade or wing over paired piecewise geodesics across respective manifolds. This procedure is consistent with the previous interpolation procedure, but is now applied independently over respective manifolds as
\begin{equation} \label{eq:prod_blade}
	X(\eta) = (\tilde{X} \circ \varphi)(\eta)(P \circ \psi)(\eta)R(\eta)+1_{n,2}\text{diag}(\bm{b}(\eta))
\end{equation}
utilizing reparametrization over SPD distances as $\psi:\eta \mapsto \ell$ and inferred (implicit) or parametrized (explicit) rotations $R(\eta)$ with paired translations $\bm{b}(\eta)$.

Truncating the Grassmannian principal basis to the first $r=4$ components (based on the decay in PGA eigenvalues), we significantly reduce the number of parameters needed to define a rich set of airfoil deformations. Consequently, we have ``learned'' a four-dimensional data-driven submanifold of airfoil undulations, $\mathcal{A}_4$, which are independent of affine deformations. New parameters are now coordinates of this four-dimensional subspace $\bm{t} \in T_{\bm{0}}\mathcal{A}_4 \cong \mathbb{R}^4$ over the tangent space at the Karcher mean (analogous origin for $\mathcal{A}_r$). This is subsequently composed with right group actions defined by normal coordinates over $S^2_{++}$, fixed average length scales $\overline{M}$, explicit or inferred parametrizations $M(\bm{\ell})$, or some combination thereof to offer a complete representation of shapes.

The right plot of Fig.~\ref{fig:scatter} shows a 2D marginal slice of the airfoil data projected onto the four-dimensional PGA basis $\mathcal{T}$---i.e., a discrete distribution of $\bm{t} \in T_{[\tilde{X}_0]}\mathcal{A}_4$. Note that this design space roughly resembles a mixture of overlapping Gaussian distributions across the diverse family of airfoils. Compared to the CST representation, such a design space is significantly easier to infer or represent in the context of AI and ML algorithms. Further, extrapolation to shapes beyond the point cloud is significantly less volatile in this framework, offering an improved notion of regularized shape deformations.

To demonstrate the improved regularization of deformations, Fig.~\ref{fig:corner_sweeps} shows four random corner-to-corner sweeps (enveloping the data by bounding hyperrectangles) through CST and PGA spaces. The PGA perturbations are constructed with fixed average scales $\overline{M}$ defining right inverse $\tilde{X}(\bm{t})\overline{M}$ for $\bm{t} \in T_{[\tilde{X}_0]}\mathcal{A}_4$. In CST space, it is difficult to define a single design space that covers the range of airfoils under consideration while allowing for smooth deformations between them---i.e., named classes correspond to largely distinct subsets of perturbed CST coefficients, and linear interpolation across these classes results in nonphysical shapes with large undulations. Conversely, all shapes generated using the proposed Grassmannian methodology result in valid airfoil designs while creating a rich design space worthy of investigation. Moreover, this data-driven approach and regularization over \eqref{eq:Gr_submanifold} mitigates undesirable oscillations and undulations in the shape, as depicted in Fig.~\ref{fig:corner_sweeps}.

\subsection{Consistent Blade Deformations}
We emphasize a powerful approach to dimension reduction for parametrizing 3D shapes from 2D sequences of discrete shapes. Namely, we introduce the concept of \emph{consistent deformations} via parallel translation over the inferred PGA domain. The deformations retain the original notion of parameter direction over the PGA domain and result in a more intuitive deformation of blade shape---only requiring 7-10 parameters in total to define rich deformations to the entire 3D blade shape.

As previously noted, current approaches to 3D blade design require significant hand-tuning of airfoils to ensure the construction of valid blade geometries without dimples or kinks. Quantitatively, this can result from poor interpolation of shape characteristics, which deform shapes differently from one cross section to the next. Our proposed approach enables the flexible design of new blades by applying consistent deformations to all airfoils in addition to affine-invariant interpolation of shapes along the span. The undesirable variations from one cross section to the next are mitigated by parametrizing consistent perturbations enabled by parallel translation.

Blade perturbations are constructed from deformations to each of the given cross-sectional airfoils in consistent directions over $\bm{t} \in T_{[\tilde{X}_0]}\mathcal{A}_4$. Having defined perturbation directions in the tangent space of the Karcher mean $[\tilde{X}_0]$, we utilize parallel transport as an isometry to smoothly translate the perturbing vector field of the shape along separate geodesics---connecting the Karcher mean to each of the LA standardized airfoils $([\tilde{X}_k])$. More generally, this could be computed for $\mathcal{G}(n,2)$ and $S^2_{++}$ using Algorithms~\ref{alg:Gr_parallel} and \ref{alg:SPD_parallel}, respectively. The result is a set of consistent directions---as equal inner products and consequently equivalent normal coordinates in the central tangent space---over ordered tangent spaces $T_{[\tilde{X}_k]}\mathcal{G}(n,2)$, centered on each of the nominal $[\tilde{X}_k]$ defining the blade (similarly for corresponding SPD matrices). An example of a consistently perturbed sequence of cross-sectional airfoils is shown in Fig.~\ref{fig:interp_blade} (right), composed with average length scales $\overline{M}$ and $\Delta(\bm{t}) = \text{vec}^{-1}(T_r\bm{t})$ as
\begin{equation}
    (\pi^{-1}\circ [\tilde{X}])(\eta_k;\bm{t}, \overline{M}) =\text{Exp}_{[\tilde{X}_k]}\left(\tau_{\parallel}(\Delta(\bm{t});[\tilde{X}_0], [\tilde{X}_k]) \right)\overline{M}.
\end{equation}
These deformed shapes are then interpolated over $(\eta_k)$ to construct various blade deformations with any level of refinement as wire-frames. Consequently, undulations in the blade can be consistently parametrized and regularized by a shared set of four parameters $\bm{t} \in T_{[\tilde{X}_0]}\mathcal{A}_4$ in the tangent space of the Karcher mean using parallel transport with Algorithm~\ref{alg:Gr_parallel} (and/or Algorithm~\ref{alg:SPD_parallel} for a chosen separable representation).

Utilizing interpolation over consistent deformations with three to six independent affine parameters $\bm{\ell}$, chosen to be spanwise constant, as $\tilde{X}(\eta; \bm{t})M(\eta; \bm{\ell})$, constitutes a full set of \textit{$7-10$ parameters that describe a rich feature space of 3D blade perturbations}---a significant reduction compared to alternative expansion representations. Of course, the more general representation \eqref{eq:polar_sep} requires distinct parallel transport over the product manifold $\mathcal{A}_4 \times \mathcal{S}_3$ to offer additional flexibility with similar implications using only seven, or generally $r+3$, total parameters. Additional generalizations can be made at the expense of introducing additional parameters---perhaps incorporating controlled spanwise variations in $\bm{t}$. However, in this most simplified context, a rich set of deformations can still be defined with very few parameters. The result is a framework with a flexible means for designers to balance the total number of parameters to achieve sought blade deformations using a more explainable and interpretable representation---generated by an explainable and interpretable nonlinear manifold of shapes.

\section{Conclusions}

The benefits of coherent shape deformations, coupled with a natural framework for interpolating 2D airfoil shapes into 3D blades and the decoupling of affine and undulation-type deformations, make Grassmannian-based shape representation a powerful tool for enabling aerodynamic design. Moreover, the proposed transformations also enable the representation of rotation- and reflection-invariant shapes over a product submanifold for the purposes of 2D design.

Transforming discrete shapes into separable tensors enables 2D deformations and designs with evidence of improved regularization against nonphysical deformations. Moreover, the transformed representation offers more visually compelling evidence that Gaussian mixture models (a common prior in AI-aided design) may be a more relevant choice of prior distribution over the Karcher (Fr\'echet) centered domain of normal coordinates. Additionally, the ability to construct consistent deformations to blade shapes via parallel transport offers a novel and intuitive regularization to dramatically reduce the total number of parameters for 3D blade design.

We have shown, through theoretical arguments and numerical demonstration, that samples drawn from a class of relevant discrete airfoil shapes converge to Grassmannian elements defined as discrete refinements with fixed reparametrization. Moreover, we have motivated continuous analogues built from the corresponding discrete shapes for future extensions of this work. Lastly, in contrast to AI-based generative shape models, our methods offer a fast and extremely lightweight approach to shape-manifold learning.

\section*{Acknowledgements}
We would like to thank our collaborators: Ryan King, Ganesh Vijayakumar, Bum Seok Lee, and James Baeder for discussions about the design requirements and constraints in wind turbine blade design. We would also like to thank Andrew Dienstfrey for supporting discussions and providing constructive technical critiques emphasizing the strengths and weaknesses of the proposed methods.

This work was authored in part by the National Renewable Energy Laboratory, operated by Alliance for Sustainable Energy, LLC, for the U.S. Department of Energy (DOE) under Contract No. DE-AC36-08GO28308. Funding partially provided by the Advanced Research Projects Agency-Energy (ARPA-E) Design Intelligence Fostering Formidable Energy Reduction and Enabling Novel Totally Impactful Advanced Technology Enhancements (DIFFERENTIATE) program, award No. 19{\textbackslash}CJ000{\textbackslash}07{\textbackslash}03. The views expressed in the article do not necessarily represent the views of the DOE or the U.S. Government. This work is U.S. Government work and not protected by U.S. copyright. A portion of this research was performed using computational resources sponsored by the Department of Energy's Office of Energy Efficiency and Renewable Energy and located at the National Renewable Energy Laboratory. Certain commercial equipment, instruments, or materials are identified in this paper in order to specify hardware utilized in numerical experiments. Such identification is not intended to imply recommendation or endorsement by NIST/NREL, nor is it intended to imply that the materials or equipment identified are necessarily the best available for the purpose.

\bibliography{bibliography}
\end{document}